\documentclass[11pt]{article}
\pdfoutput=1
\usepackage{etoolbox}
\providebool{twocol}
\setbool{twocol}{false}

\usepackage{amsmath,amssymb,amsthm,amsfonts}
\usepackage[ttscale=0.8]{libertine}
\usepackage[T1]{fontenc}
\usepackage{xspace}
\usepackage{graphicx}
\usepackage{algorithm,algorithmicx,algpseudocode}
\usepackage[margin=1in,footskip=24pt]{geometry}
\usepackage{tikz}
\usepackage{caption2}
\usepackage{graphicx}

\usetikzlibrary{arrows}

\usepackage{hyperref}
\hypersetup{
  colorlinks=true,
  linkcolor=blue,
  linktoc=all,
  citecolor=green!30!black,
  bookmarksnumbered=true
}
\usepackage{cleveref}

\newlength\abovesectionskip
\newlength\belowsectionskip
\def\sectionfont{\normalfont\Large\bfseries}
\newlength\abovesubsectionskip
\newlength\belowsubsectionskip
\def\subsectionfont{\normalfont\large\bfseries}
\newlength\abovesubsubsectionskip
\newlength\belowsubsubsectionskip
\def\subsubsectionfont{\normalfont\normalsize\bfseries}
\newlength\aboveparagraphskip
\newlength\belowparagraphskip
\def\paragraphfont{\normalfont\normalsize\bfseries}
\makeatletter
\def\section{\@startsection{section}{1}{\z@}{-\abovesectionskip}%
               {\belowsectionskip}{\sectionfont}}
\def\subsection{\@startsection{subsection}{2}{\z@}{-\abovesubsectionskip}%
                  {\belowsubsectionskip}{\subsectionfont}}
\def\subsubsection{\@startsection{subsubsection}{3}{\z@}%
                     {-\abovesubsubsectionskip}{\belowsubsubsectionskip}%
                     {\subsubsectionfont}}
\def\paragraph{\@startsection{paragraph}{4}{\z@}{-\aboveparagraphskip}%
                 {-\belowparagraphskip}{\paragraphfont}}
\makeatother

\makeatletter
\renewenvironment{align*}{%
  \abovedisplayskip 5pt plus 1pt%
  \belowdisplayskip 5pt plus 1pt%
  \start@align\@ne\st@rredtrue\m@ne
}{%
  \endalign
}
\let\stdequation\equation
\renewcommand*\equation{%
  \abovedisplayskip 5pt plus 1pt%
  \belowdisplayskip 5pt plus 1pt%
  \stdequation}
\DeclareRobustCommand{\[}{
  \abovedisplayskip 5pt plus 1pt%
  \belowdisplayskip 5pt plus 1pt%
  \begin{equation*}
}
\makeatother

\renewenvironment{itemize}{
    \begin{list}{$\bullet$}{
        \setlength{\labelsep}{6pt}\setlength{\itemindent}{0mm}\setlength{\labelwidth}{3mm}
        \setlength{\leftmargin}{30pt}
        \setlength{\itemsep}{2pt}\setlength{\parsep}{0mm}
        \setlength{\topsep}{3pt}\setlength{\listparindent}{0pt}
    }
}
{
    \end{list}
}

\allowdisplaybreaks[3]      %

    \newtheoremstyle{mythmstyle}
      {6pt}   %
      {6pt}   %
      {}            %
      {}            %
      {\bfseries}   %
      {. }          %
      {2.5pt}       %
      {\thmname{#1}\thmnumber{ #2}\thmnote{ \normalfont (#3)}}   %
    \theoremstyle{mythmstyle}
    \newtheorem{theorem}{Theorem}[section]\numberwithin{equation}{section}
    \newtheorem{corollary}[theorem]{Corollary}
    \newtheorem{claim}[theorem]{Claim}
    \newtheorem{definition}[theorem]{Definition}
    \newtheorem{lemma}[theorem]{Lemma}

    \newtheorem{fact}[theorem]{Fact}
    \newcommand{\abs}[1]{\lvert #1 \rvert}

    \newcommand{\defeq}{\,:=\,}                                     %
    \newcommand{\ceil}[1]{\left\lceil #1 \right\rceil}
    \newcommand{\setst}[2]{\left\{\; #1 \,:\, #2 \;\right\}}        %
    \newcommand{\smallsum}[2]{{\textstyle \sum_{#1}^{#2}}}
    \newcommand{\smallprod}[2]{{\textstyle \prod_{#1}^{#2}}}
    \newcommand{\comment}[1]{}

    \newcommand{\bC}{\mathbb{C}}
    \newcommand{\bR}{\mathbb{R}}
    \newcommand{\bZ}{\mathbb{Z}}
    \newcommand{\cC}{\mathcal{C}}
    \newcommand{\cE}{\mathcal{E}}
    \newcommand{\cL}{\mathcal{L}}
    \newcommand{\cS}{\mathcal{S}}
    \newcommand{\cX}{\mathcal{X}}

    \newcommand{\ClaimName}[1]{\label{clm:#1}}
    
    \newcommand{\CorollaryName}[1]{\label{cor:#1}}
    
    \newcommand{\EquationName}[1]{\label{eq:#1}\text{}}

    \newcommand{\LemmaName}[1]{\label{lem:#1}}
    
    \newcommand{\TheoremName}[1]{\label{thm:#1}}
    \newcommand{\SectionName}[1]{\label{sec:#1}}
    
    \newcommand{\Algorithm}[1]{Algorithm~\ref{alg:#1}}
    
    \newcommand{\Claim}[1]{Claim~\ref{clm:#1}}
    
    \newcommand{\Corollary}[1]{Corollary~\ref{cor:#1}}

    \newcommand{\Lemma}[1]{Lemma~\ref{lem:#1}}
    
    \newcommand{\Section}[1]{Section~\ref{sec:#1}}
    
    \newcommand{\Theorem}[1]{Theorem~\ref{thm:#1}}

\usepackage[english]{babel}%
\usepackage[utf8]{inputenc}
\usepackage[noadjust]{cite}
\usepackage{etoolbox}

\newbool{garamond}
\setbool{garamond}{false}

\notbool{garamond}{}{

}

\usepackage{bm}
\usepackage{ifthen}
\usepackage{letltxmacro}

\usepackage{letltxmacro}

\newcommand{\arxiv}[1]{\href{https://arxiv.org/abs/#1}{arXiv:#1}}

\usepackage{nicefrac}
\usepackage{xstring}
\usepackage{pgf}
\usepackage{tikz}
\usetikzlibrary{arrows,automata}
\usepackage{marginnote}
\usepackage{todo}

\ifbool{garamond}{
  \usepackage[urw-garamond]{mathdesign}
}{
  \usepackage{amssymb, amsfonts}
}

\usepackage[T1]{fontenc}

\renewcommand{\mathbf}[1]{\bm{#1}}

\usepackage{graphicx}
\usepackage{booktabs,siunitx}

\newtheorem{observation}[theorem]{Observation}
\newtheorem{question}{Question}
\newtheorem*{question*}{Question}
\newtheorem{remark}{Remark}[section]

\renewcommand{\abs}[1]{\ensuremath{\left|#1\right|}}

\newcommand{\diff}[2]{\frac{\text{d}#1}{\text{d}#2}}

\usepackage{mathtools}

\newcommand{\E}[2][]{\ensuremath{\mathbb{E}_{#1}\insq{#2}}}

\newcommand{\inb}[1]{\left\{#1\right\}}
\newcommand{\inp}[1]{\left(#1\right)}
\newcommand{\insq}[1]{\left[#1\right]}


\newcommand{\nfrac}[3][]{\nicefrac[#1]{#2}{#3}}

\newcommand{\Z}[0]{\ensuremath{\mathbb{Z}}}

\renewcommand{\emptyset}[0]{\varnothing}
\newcommand{\poly}[1]{\ensuremath{\mathop{\mathrm{poly}}\inp{#1}}}

\renewcommand{\vec}[1]{\mathbf{#1}}

\newcommand{\CC}{\mathbb C}

\makeatletter
\newlength{\trianglerightwidth}
\algnewcommand{\LineCommentCont}[1]{\Statex \hskip\ALG@thistlm%
  \parbox[t]{\dimexpr\linewidth-\ALG@thistlm}
{\leftskip=0pt
  \hangindent=0pt
  \hangafter=1%
  \strut\makebox[\algorithmicindent][l]{$\triangleright$}#1\strut}
  } %

\algnewcommand{\LineCommentContNoInd}[1]{\Statex %
  \parbox[t]{\dimexpr\linewidth}
{\leftskip=0pt
  \hangindent=20pt
  \hangafter=1%
  \strut\makebox[\algorithmicindent][l]{$\triangleright$}#1\strut}
  }

\algnewcommand{\MyState}[1]{\State
\parbox[t]{\dimexpr\linewidth-\ALG@thistlm}{\hangindent=0pt\strut\hangafter=1#1\strut}}
\makeatother

\newcommand{\set}[1]{\left \{ #1 \right \}}                     %
\newcommand{\intersect}{\cap}
\newcommand{\veczero}{\mathbf{0}}
\newcommand{\plusminus}{\pm}
\newcommand{\Intersect}{\bigcap}
\newcommand{\cR}{\mathcal{R}}
\newcommand{\qdown}{\breve{q}}
\newcommand{\Ind}{\operatorname{Ind}}

\newcommand{\children}{\cC}
\newcommand{\ra}{\varrho}
\newcommand{\raPrime}[1]{\varrho'_{#1}}
\newcommand{\depth}{\delta}
\newcommand{\SAWTree}{\ensuremath{T_{\mathrm{SAW}}}}
\newcommand{\SharpPHard}{\#P-hard\xspace}

\newcommand{\DTIME}{\text{DTIME}}

\begin{document}

\title{Computing the Independence Polynomial: \\ from the Tree Threshold down to the Roots}

\newcommand{\calgrant}{NSF grant CCF-1319745}
\author{Nicholas J. A. Harvey\thanks{Email:
    \texttt{nickhar@cs.ubc.ca}. University of British Columbia.} \and
     Piyush Srivastava\thanks{Email:
       \texttt{piyush.srivastava@tifr.res.in}.  Tata Institute of
       Fundamental Research.} \and
  Jan Vondr\'ak\thanks{Email: \texttt{jvondrak@stanford.edu}.  Stanford University.}}
\date{}

\maketitle
\thispagestyle{empty}

\begin{abstract}

We study an algorithm for approximating the multivariate independence polynomial $Z(\vec{z})$, 
with negative and complex arguments.
While the focus so far has been mostly on computing combinatorial polynomials restricted to
the univariate positive setting (with seminal results for the independence polynomial by Weitz (2006) and Sly (2010)),
the independence polynomial with \emph{negative} or \emph{complex} arguments
has strong connections to combinatorics and to statistical physics.
The independence polynomial with negative arguments, $Z(-\vec{p})$, determines the {\em Shearer region},
the maximal region of probabilities to which the Lov\'{a}sz Local Lemma (LLL) can be extended (Shearer 1985).
In statistical physics, complex zeros of the independence polynomial relate to existence of phase transitions.

Our main result is a deterministic algorithm to compute approximately the independence polynomial
in any root-free complex polydisc centered at the origin.
More precisely,  we can $(1+\epsilon)$-approximate the independence polynomial $Z(\vec{z})$
for an $n$-vertex graph of degree at most $d$, for any complex vector $\vec{z}$ such that
$Z(\vec{z'}) \neq 0$ for $|z'_i| \leq (1+\alpha) |z_i|$,
in running time $(\frac{n}{\epsilon \alpha})^{O(\log(d)/ \sqrt{\alpha})}$.  Our result also
extends to graphs of unbounded degree that have a bounded connective constant.
Our algorithm is essentially the same as Weitz's algorithm for positive parameters 
up to the tree uniqueness threshold.
The core of the analysis is a novel multivariate form of the correlation decay technique,
which can handle non-uniform complex parameters.
In summary,  we provide a unifying algorithm
for all known regions where $Z(\vec{z})$ is approximately computable.
In particular, in the univariate real setting our work implies that Weitz's algorithm works in an interval
between two critical points $(-\lambda'_c(d), \lambda_c(d))$, and outside of this interval
an approximation of $Z(\lambda)$ is known to be NP-hard.

As an application, we provide an algorithm to test membership in
Shearer's region within a multiplicative error of $1+\alpha$, in
running time $(n/\alpha)^{O(\sqrt{n/\alpha} \log d)}$. We also give a
deterministic algorithm for Shearer's lemma (extending the LLL) with
$n$ events on $m$ independent variables under slack $\alpha$, with
running time $(nm/\alpha)^{O(\sqrt{m/\alpha} \log d)}$.

On the hardness side, we prove that evaluating $Z(\vec{z})$ at an arbitrary point in Shearer's region,
and testing membership in Shearer's region, are $\#P$-hard problems.
For Weitz's correlation decay technique in the negative regime, we show that the $1/\sqrt{\alpha}$ dependence
in the exponent is optimal.\footnote{An earlier version of this paper gave an algorithm with $1/\alpha$ dependence in the exponent,
which would lead to trivial exponential running time in our applications.}

%

%

%

%


\end{abstract}

\newpage
\thispagestyle{empty}
\tableofcontents

\newpage\pagestyle{plain}\setcounter{page}{1}

\section{Introduction}
\label{sec:introduction}

The \emph{independence polynomial} is the generating function of
independent sets of a graph.  Formally, given a graph
$G = (V, E)$, and a vector $\vec{x} = \inp{x_v}_{v \in V}$ of
\emph{vertex activities}, it is the multi-linear polynomial
\begin{displaymath}
  Z_G(\vec{x}) = \sum_{I\text{ indep. in }G} \: \prod_{v \in I} x_v.
\end{displaymath}
Aside from its natural importance in combinatorics as a generating
function, the independence polynomial has also been studied
extensively in statistical physics where it arises as the
\emph{partition function} of the \emph{hard core lattice gas},
which has been used as a model of adsorption.  In
both settings, the partition function and its derivatives encode
important properties of the model.  For example, in the combinatorial
setting, $Z_G$ encodes a weighted count of the
independent sets, while the derivatives of $\log Z_G$ encode
relevant average quantities, such as the mean size of an independent set.
As such, much effort has gone into
understanding the complexity of computing $Z_G$.
The exact evaluation of the independence polynomial at non-trivial evaluation
points turns out to be \#P-hard~\cite{val79b}. 
As for approximate computation, the problem is well studied in the setting where
the activities are \emph{positive} and \emph{real valued}.  In this
setting, the problem has served to highlight some of the tightest
known connections between phase transitions and computational
complexity: we will discuss this line of work in more detail below.

In this paper, we are concerned instead with the problem of
approximately computing the independence polynomial at possibly
negative and even complex valued vertex activities.  The interest in
studying partition functions at complex values of the activities
originally comes from statistical mechanics, where there is a paradigm
of studying phase transitions in terms of the analyticity of
$\log Z_G$.  This paradigm has led to the question of characterizing
regions of the complex plane where the partition function is
non-zero~\cite{leeyan52,leeyan52b}. %
Inspired by previous connections between statistical physics
and computation complexity, a natural question is:
does the maximum radius around the origin within which $\log Z_G$ is analytic
(i.e., within which $Z_G$ has no roots)
correspond to a transition in the computational complexity of computing $Z_G$?
As we discuss below, the answer is yes.
  
A second motivation for studying the
independence polynomial at complex activities comes from a
delightful connection between combinatorics and statistical mechanics
that arose in the work of Shearer~\cite{Shearer} and Scott and
Sokal~\cite{ScottSokal} on the Lov\'asz Local Lemma (LLL).
In particular, the largest region of parameters in which the LLL applies is 
the maximal connected region of the negative orthant within which $Z_G$ has no roots.
As we discuss below, an algorithm for approximating $Z_G$
at negative activities has several algorithmic applications relating to the
LLL, including testing whether the hypotheses are satisfied,
as well as giving a constructive proof of the LLL itself.

\subsection{Our results}
\label{sec:our-results}

Before stating our results, let us define the \emph{Shearer region} for graph $G$ to be
\ifbool{twocol}{
\begin{multline}
  \label{eq:8}
\cS ~=~ \cS_G ~\defeq~ \\ \setst{ \vec{p} \in [0,1]^V }{
    Z_G\inp{\vec{z}} \neq 0 ~~\forall \vec{z} \in \bC^V ~\text{s.t.}~ \abs{\vec{z}} \leq \vec{p}
},
\end{multline}
}{
\begin{equation}
  \label{eq:8}
\cS ~=~ \cS_G ~\defeq~ \setst{ \vec{p} \in [0,1]^V }{
    Z_G\inp{\vec{z}} \neq 0 ~~\forall \vec{z} \in \bC^V ~\text{s.t.}~ \abs{\vec{z}} \leq \vec{p}
},
\end{equation}
}
which describes the radii of polydiscs within which $Z_G$ has no roots.
Here, $\abs{\vec{z}}$ means coordinate-wise magnitude,
and $\leq$ also applies coordinate-wise.  It can be shown that
$\cS_G$ is an open set.

Our main result is a fully polynomial time approximation scheme
(FPTAS) for $Z_G(\vec{z})$ when $\vec{z}$ is a vector of possibly
complex activities for which the vector $\abs{\vec{z}}$
of their magnitudes lies in the Shearer region.  
\begin{theorem}[\textbf{FPTAS for $Z_G$}]
  \TheoremName{mainShearerIntro}
  Let $G$ be an $n$-vertex graph with maximum degree $d$.
  Suppose that $\alpha, \epsilon \in (0,1]$,
  and that $\vec{z} \in \bC^V$ satisfies $(1+\alpha) \cdot \abs{\vec{z}} \in \cS_G$.
  Then a $( 1+\epsilon )$-approximation to $Z_G(\vec{z})$ can be
  computed in time $\big( \frac{n}{\epsilon \alpha} \big)^{O(\log(d)/\sqrt{\alpha})}$.
\end{theorem}
As the set $\cS_G$ is somewhat mysterious, it is instructive to consider the following
univariate corollary.
\newcommand{\GraphThreshold}{\lambda_G}
\newcommand{\ShearerThreshold}{\lambda'_c}

\ifbool{twocol}{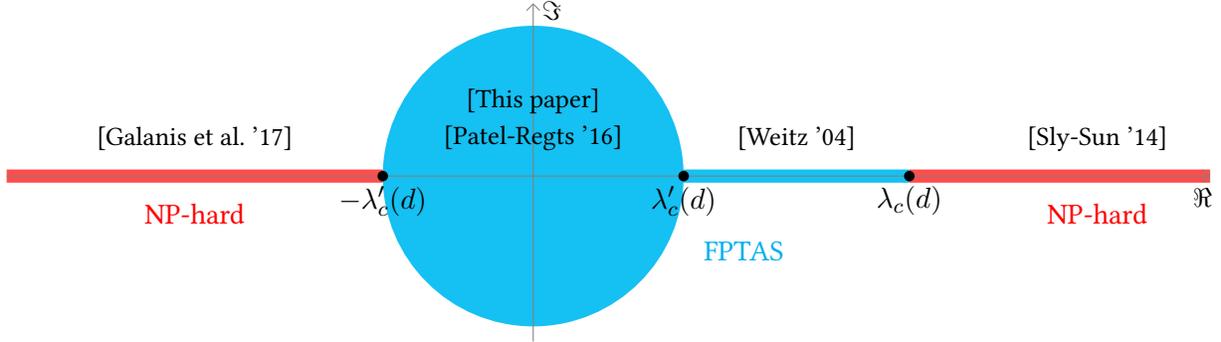
\begin{figure*}[t!]}{\begin{figure}[t!]}
\begin{tikzpicture}

\coordinate[] (O) at (0,0);
\coordinate[] (A) at (-7,0);
\coordinate[] (B) at (9,0);

\fill [color=cyan!70] (O) circle (2);

\coordinate[] (C) at (0,-2.2);
\coordinate[] (D) at (0,2.3);
\coordinate[label=below:$-\lambda'_c(d)$] (E) at (-2,0);
\coordinate[label=below:$\lambda'_c(d)$] (F) at (2,0);
\coordinate[label=below:$\lambda_c(d)$] (G) at (5,0);

\draw [color=cyan!70, line width = 5pt] (O) -- (G);
\draw [color=red!70, line width = 5pt] (G) -- (B);
\draw [color=red!70, line width = 5pt] (A) -- (E);

\draw [arrows={-angle 90},color=black!50] (A) -- (B);
\draw [arrows={-angle 90},color=black!50] (C) -- (D);
\draw (8.9,-0.28) node {\small $\Re$};
\draw (0.245,2.2) node {\small $\Im$};

\fill (E) circle (2pt);
\fill (F) circle (2pt);
\fill (G) circle (2pt);

\draw [color=cyan] (2.8,-1) node {FPTAS};
\draw [color=red] (-4.5,-0.5) node {NP-hard};
\draw [color=red] (7.5,-0.5) node {NP-hard};

\draw (0,1) node {\small [This paper]};
\draw (0,0.5) node {\small [Patel-Regts '16]};
\draw (3.5,0.5) node {\small [Weitz '04]};
\draw (-4.5,0.5) node {\small [Galanis et al.~'17]};
\draw (7.5,0.5) node {\small [Sly-Sun '14]};
\end{tikzpicture}
\caption{{\footnotesize Summary of results for computation of
    $Z_G(z \vec{1})$ in the complex univariate setting, as a function of the degree $d$.
    Here $\lambda'_c = \frac{(d-1)^{d-1}}{d^d} \searrow \frac{1}{ed}$ and
         $\lambda_c  = \frac{(d-1)^{d-1}}{(d-2)^d} \searrow \frac{e}{d}$.
    Note that a major difference of this work from Patel-Regts~\cite{PatelRegts}
    that is not captured by this figure is that our running time has a
    significantly better dependence on distance from the boundary of Shearer's
    region, which is crucial in our applications.
  }\label{fig:1} }
\ifbool{twocol}{\end{figure*}}{\end{figure}}


\begin{corollary}[\textbf{FPTAS for the univariate case}]
  \CorollaryName{univariate} Let $G$ be an $n$-vertex graph with
  maximum degree $d$.
  Define $\GraphThreshold = \min \{\, \abs{z} \,:\, z \in \bC ,\, Z_G( z \vec{1} ) = 0 \,\}$.  
  Let $\alpha, \epsilon \in (0,1]$ and $z \in \bC$ satisfy
  $(1+\alpha) \abs{z} \leq \GraphThreshold$.  Then a
  $( 1+\epsilon )$-approximation to $Z_G(z \vec{1})$ can be deterministically computed in time
  $\big( \frac{n}{\epsilon \alpha} \big)^{O((1/\sqrt{\alpha})\cdot\log(d))}$.
\end{corollary}

\paragraph{Remark: Region of applicability.}
In order to understand $\GraphThreshold$ it is helpful to consider a bound
that depends only on the degree $d$.
Define $\ShearerThreshold(d)$ to be the minimum of $\GraphThreshold$ over all graphs of maximum degree
$d$. Then it is known \cite{Shearer} that $\ShearerThreshold(1) = 1/2$ and $\ShearerThreshold(d) =
\frac{(d-1)^{d-1}}{d^d}$ for $d \geq 2$; the minimum is achieved by the infinite $d$-regular tree.
So $\ShearerThreshold(d)$ is the threshold, depending only on $d$, that determines the region of applicability
of our algorithm. This is no accident:
approximating $Z_G(z \vec{1})$ for real $z < -\ShearerThreshold$ has recently been shown
to be NP-hard by Galanis, Goldberg and
\v{S}tefankovi\v{c}~\cite{galanis_inapproximability_2016}, showing
that \Corollary{univariate} %
has the
tightest possible range of applicability on the negative real line (i.e., the Shearer
region). %
Thus, a phase transition in the computational complexity of the problem occurs right at the
boundary of the region within which $Z_G$ is guaranteed to have no roots.
As \Cref{fig:1} shows, we now have a complete picture of the computational complexity of $Z_G$
in the real univariate case, as a function of $d$.

\begin{figure}[b!]
\begin{center}\includegraphics[width=130pt,trim=0.05in 4.2in 7.75in 1.3in]{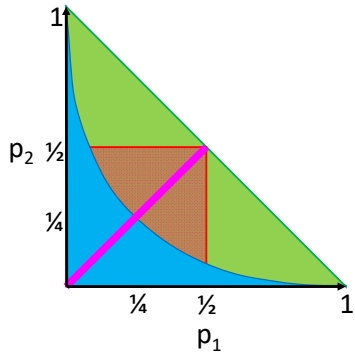}\end{center}
\caption{{\footnotesize For illustrative purposes, let us consider the graph $G = K_2$;
here $d=1$ and $\ShearerThreshold(d)=1/2$. %
Our algorithm applies throughout Shearer's region $\cS_G$, shown as the green triangular region.
The pink line segment is the restriction of $\cS_G$ for the univariate function $Z_G(z \vec{1})$.
The work of Patel-Regts describes an algorithm focused on the univariate case,
but they mention \cite[pp.~13]{PatelRegts} that it can be generalized to
all points dominated by $\ShearerThreshold$ (the red square region).
The blue region, defined as $\cL$ in \Cref{app:LLL}, is where the original
LLL~\cite{ErdosLovasz,Spencer77} applies.  %
}\label{fig:2} }
\end{figure}

\paragraph{Remark: Dependence on Slack.}
An important feature of \Cref{thm:mainShearerIntro} is that although
the running time degrades as the input vector $\vec{p}$ approaches the
boundary of the Shearer region $\cS$, the degradation is only
\emph{sub-exponential} in $\frac{1}{\alpha}$ (being exponential in
$\frac{1}{\sqrt{\alpha}}$) where $\alpha$ is the \emph{slack}
parameter that measures the distance to the boundary.  This is in
contrast to an earlier manuscript of the present
paper~\cite{harvey2016computing} and the concurrent paper of Patel and
Regts~\cite{PatelRegts}, which (using different methods), obtained an
FPTAS whose running time is actually \emph{exponential} in
$\frac{1}{\alpha}$.  We describe the new ideas required to get this
better dependence on $\alpha$ in \Cref{sec:corr-decay-with-2}, and
remark on the barriers to improving this dependence towards the end of
this subsection.

The importance of a sub-exponential dependence on the slack is that
for some applications it is imperative to approximate the
independence polynomial at points that are extremely close to the
boundary of the Shearer region and have slack at most $\Theta(1/n)$.
We present two such applications here, for both of which we are able
to obtain sub-exponential time algorithms, and for both of which the
earlier results~\cite{harvey2016computing,PatelRegts} only give
exponential time algorithms.

\paragraph{Remark: Connective constant.}
  \Cref{thm:mainShearerIntro} extends to graphs of unbounded maximum
  degree that have a bounded \emph{connective
    constant}~\cite{madras96:_self_avoid_walk,
    hammersley_percolation_1957, SSSY15, SSY13}. See
  \ifbool{twocol}{the full version}{\Cref{sec:extens-graphs-bound}} for the details of this extension.

\paragraph{Application 1: Testing membership in Shearer region.}
Physicists have studied the univariate threshold $\GraphThreshold$ for specific graphs,
as this determines the region within which there are no phase transitions \cite{leeyan52}.
For example, to understand phase transitions in $\bZ^2$,
researchers have performed numerical computations on finite graphs to estimate the exact value
$\GraphThreshold(\bZ^2)$.
(See, e.g., \cite{HolroydLiggett} \cite[Section 8.4]{ScottSokal} \cite{Todo}.)
Computations have shown that $\GraphThreshold(\bZ^2) \leq 1/8$ (rigorous)
and $\GraphThreshold(\bZ^2) = 0.119,338,881,88(1)$ (non-rigorous).

Our first application is an algorithm to test whether a given 
vector $\vec{p}$ lies in the Shearer region, up to accuracy $\alpha$. 
This can be used to compute bounds on $\GraphThreshold$,
and could potentially be useful for physicists.

\begin{theorem}
  Given a graph $G$, $\vec{p} \in (0,1)^V$, and $\alpha \in (0,1]$,
  there exists a deterministic algorithm which, in running time
  $(n/\alpha)^{O(\sqrt{n / \alpha} \log d)}$ decides whether
  $\vec{p} \in {\cal S}_G$ or $(1+\alpha) \vec{p} \notin {\cal S}_G$.
\end{theorem}

This algorithm uses the FPTAS of \Theorem{mainShearerIntro} in a black box fashion,
calling it at points in $\mathcal{S}$ that may have slack $O(1/n)$.
Replacing the black box by an algorithm that had an exponential
dependence on the slack would give an algorithm with only a trivial
exponential time guarantee on its run-time.  We note also that the
testing membership in $\cS$ is \#P-hard when $\alpha$ is exponentially small.
(See \Cref{app:hardness} for a precise statement of this hardness result.)

\paragraph{Application 2: Constructive algorithm for the
    Lov\'asz Local Lemma by polynomial evaluation.}
The LLL is a tool in combinatorics giving conditions ensuring that
it is possible to avoid certain bad events $\cE_1,\ldots,\cE_n$.
(For readers unfamiliar with the LLL, a statement is provided in \Cref{app:LLL}.)
Although the LLL guarantees that there exists a point in $\Intersect_{i=1}^n \overline{\cE_i}$,
it provides no hint on how to find such a point.
For decades, algorithmically constructing such a point was a major research challenge,
though over the past 10 years dramatic progress has been made.
Any such algorithm must necessarily make some assumptions on the probability space,
the most common being the ``variable model'' used by \cite{MoserTardos}.
All previous algorithms have been based on the idea of
randomly sampling variables followed by brute-force search \cite{Beck91,Alon91},
or random resampling
\cite{Srinivasan08,MoserConf09,MoserTardos,Kolipaka,Achlioptas,HarrisThesis,HVConf},
or derandomizations of those ideas \cite{Beck91,Alon91,MoserConf09,MoserTardos,Chandrasekaran13}.

We develop a completely new algorithmic approach to the LLL
in the variable model.
The previous randomized algorithms can be viewed as generating a sequence of
\textit{infeasible, integral} solutions;
at each step, they resample one of the bad events and hopefully move closer to feasibility.
(The previous deterministic algorithms are derandomizations of this approach.)
In contrast, our new algorithm generates a sequence of \textit{feasible, fractional} solutions;
at each step, it fixes the value of one of the variables while preserving feasibility in $\cS_G$.
The value of the polynomial $Z_G(\vec{z})$ is used to determine membership in $\cS_G$.
Thus, our algorithm can be viewed as a \textit{rounding algorithm} for the LLL,
and the value of $Z_G(\vec{z})$ can be viewed as a \textit{pessimistic estimator}
for the probability of $\Intersect_{i=1}^n \overline{\cE_i}$.
To compute $Z_G(\vec{z})$, our algorithm uses (as a black box) our deterministic FPTAS for
evaluating the independence polynomial with negative activities
and slack $\Omega(1/m)$, where $m$ is the number of variables.

\begin{theorem}
Consider an LLL scenario in the variable model (as in \Cref{app:VarMod}):
$\mu_{\vec{z}}$ is the product distribution on $\set{0,1}^m$ with expectation $\vec{z}$,
$G$ is the dependency graph for events $\cE_1,\ldots,\cE_n$, and $p_i = \mu_{\vec{z}}(\cE_i)$.
There is a deterministic algorithm that takes as input
a description of the events $\cE_1,\ldots,\cE_n$, a vector $\vec{z} \in [0,1]^m$,
and a parameter $\alpha \in (0,1]$ such that $(1+\alpha) \cdot \vec{p}(\vec{z}) \in \cS_G$.
The algorithm runs for time $(nm/\alpha)^{O(\log(d) \sqrt{m/\alpha})}$ 
and outputs a point in $\Intersect_{i=1}^n \overline{\cE_i}$.
\end{theorem}

This algorithm uses our FPTAS from \Cref{thm:mainShearerIntro} as a black box.
Note that our algorithm runs in subexponential time,
so, as of now, its runtime is not competitive with the state of the art
deterministic algorithms for the LLL~\cite{Chandrasekaran13}.
Nevertheless, prior to our work there was essentially only one known algorithmic
technique known for the LLL: the witness tree technique originating with Beck \cite{Beck91}.
Our work provides the only other known technique that gives an algorithm for the LLL better than brute-force.
The fact that our algorithm is slow is only because the best known implementation of the black box
(i.e., \Cref{thm:mainShearerIntro}) has a running time that depends
sub-exponentially on the slack.  The algorithm thus points to a new
intriguing connection between approximate counting and algorithmic
versions of the LLL, and suggests the open question of finding the
optimal dependence on the slack $\sqrt{\alpha}$ in
\Cref{thm:mainShearerIntro}.

\paragraph{Dependence on the ``slack parameter'' $\alpha$.} The
discussion following the two applications above suggests that the
question of the optimal dependence on the slack $\alpha$ in
\Cref{thm:mainShearerIntro} is of importance for further exploration
of the connection between approximate counting and the LLL.  While we
cannot yet provide a complete answer to this question, we conclude
this section with a couple of our results that address this point.
Our first result in this direction shows that some dependence on the
slack parameter is inevitable. (See \Cref{sec:approxEvalMemb} for a proof).
\begin{theorem}[\textbf{Necessity of slack}]
  \TheoremName{ShearerHardnessIntro}
If there is an algorithm to estimate $Z_G(-\vec{p})$, assuming $(1+\alpha) \vec{p} \in
\mathcal{S}$, within a $\poly{n}$ multiplicative factor in running time $(n \log
\frac{1}{\alpha})^{O(\log n)}$ then $\text{\#P} \subseteq \DTIME(n^{O(\log n)})$.
\end{theorem}

However, this hardness result, while applying to all algorithmic
approaches, only provides a weak lower bound on what can be achieved.
We do not yet have any stronger general lower bounds, but our second
result, described in detail in \Cref{sec:optim-decay-rate}, presents
evidence that the dependence on $1/\sqrt{\alpha}$ in
\Cref{thm:mainShearerIntro} is optimal for the techniques used in our
paper.  Nevertheless, it does not preclude the possibility that other
approximate counting techniques could substantially improve upon
\Cref{thm:mainShearerIntro}.  We discuss some related future
directions in \Cref{sec:concl}.

\ifbool{twocol}{\newpage}{}

\subsection{Related work}
\label{sec:relat-work-techn}

As discussed above, the exact computation of the independence
polynomial turns out to be \#P-hard.  This is a fate shared by the
partition functions of several other ``spin systems'' (e.g., the Ising
model) in statistical physics, and by now there is extensive work on
the complexity theoretic classification of partition functions in terms
of dichotomy theorems (see e.g., \cite{cai_complexity_2012}).

The approximation problem for a univariate partition function
with a positive real argument is also well
studied and has strong connections with phase
transitions in statistical mechanics.  In two seminal papers,
Weitz~\cite{Weitz} and Sly~\cite{Sly2010CompTransition} (see
also~\cite{sly12,Vigoda-hard-core-11,galanis_inapproximability_2015}) showed that there exists a
critical value $\lambda_c(d)$ such when $\lambda < \lambda_c(d)$,
there is an FPTAS for the partition function $Z_G(\lambda)$ on graphs
of maximum degree $d$, while for $\lambda > \lambda_c(d)$ close to the threshold,
approximating $Z_G(\lambda)$ on $d$-regular graphs is NP-hard under randomized
reductions. (Sly and Sun \cite{sly12} extended the hardness result to any $\lambda > \lambda_c(d)$.)

The approach for our FPTAS builds upon the correlation decay technique
pioneered by Weitz, which has since inspired several results in
approximate counting (see, e.g., \cite{li_correlation_2011,
  gamarnik_correlation_2007, efthymiou2016convergence,
  liu_fptas_2015, SSSY15, LL15, sinclair_approximation_2012,
  bayati_simple_2007}).
Unlike previous work,
where the partition function has positive activities 
and induces a probability distribution on the underlying structures,
our emphasis is on negative and complex activities.
It turns out that Weitz's proof can be easily modified to handle a
\emph{univariate} independence polynomial with a \emph{negative} (and indeed, \emph{complex}) parameter $z$ satisfying
$\abs{z} < \ShearerThreshold(d) = \frac{(d-1)^{d-1}}{d^d}$,
analogous to the $\lambda < \lambda_c(d)$ condition mentioned above;
this observation appears in~\cite{SrivastavaNote}.
Our work considers a much more general scenario:
the \emph{multivariate} independence polynomial 
under a global condition incorporating all vertex activities
(i.e., the set $\cS$).  %
This yields a result for the univariate case \emph{stronger} than \cite{SrivastavaNote},
as our threshold $\GraphThreshold$ in \Corollary{univariate} depends on $G$ not just on $d$.

Starting with a paper of Barvinok~\cite{barvinok_computing_2015}, a
different approach to approximating partition functions in their
zero-free regions has emerged. Here, the analyticity of $\log Z$ in
the zero-free region of $Z$ is used to provide an additive
approximation to $\log Z$ (which translates to a multiplicative
approximation for $Z$) via a Taylor expansion truncated at an
appropriate degree.  While this method has by now been applied to
several classes of partition
functions~\cite{Barvinok15,BarvinokSoberon16a,BarvinokSoberon16b,Regts},
the resulting algorithms had turned out to be quasi-polynomial in the
earlier applications because of the lack of a method to efficiently
compute coefficients of terms of degree $\Omega(\log n)$ in the Taylor
expansion of $\log Z$ (which, in the case of the hard core model,
correspond to $\Omega(\log n)$-wise correlations among vertices in a
random independent set).  In work that was circulated concurrently
with an earlier manuscript~\cite{harvey2016computing} of this paper,
Patel and Regts~\cite{PatelRegts} showed that for a class of models,
these coefficients could be computed in polynomial time on bounded
degree graphs, and as a consequence obtained an FPTAS for some
partition functions in the region of analyticity of their
logarithms. This included the \emph{univariate} independence
polynomial in bounded-degree graphs of degree $d$ when the activity
$\lambda$ satisfies $\abs{\lambda} < \ShearerThreshold$.  While
their particular result for the univariate independence polynomial
seems to be implied by the observations in \cite{SrivastavaNote}
pointed out above, their technique applies also to other models.  We
note however that the present work has advantages over the result of
Patel and Regts in two qualitative aspects which are both crucial for
our applications.

First, as mentioned above, the running time of our FPTAS is
sub-exponential in $1/\alpha$ when the input activity (or more
generally, the input probability vector in the multivariate case) has
slack $\alpha$, whereas their algorithm has an exponential dependence
on $1/\alpha$.  Indeed, it appears that this exponential dependence on
$1/\alpha$ is intrinsic to their approach since the rate of
convergence of the power series they use for approximating $\log Z$ is
exactly $1 - \alpha$ for an activity that has slack $\alpha$, so that
the number of terms of the series that need to be evaluated for an
additive $\epsilon/n$-approximation to $\log Z$ (which corresponds to
a $(1 \pm \Theta(\epsilon/n))$ multiplicative approximation for $Z$)
is
$\Omega(\log_{1/(1-\alpha)}\log(n/\epsilon)) =
\Omega\inp{\frac{\log(n/\epsilon)}{\alpha}}$.  Since the complexity of
computing the $k$th term of this series in their framework is
$\Omega(d^k)$, this leads to a run time that is
$\inp{\nfrac{n}{\epsilon}}^{\Omega\inp{(1/\alpha)\cdot\log(d)}}$ (our
algorithm, in contrast will run in time
$\inp{\nfrac{n}{\alpha\epsilon}}^{O\inp{(1/\sqrt{\alpha})\cdot\log(d)}}$).
As discussed above, this improvement over Patel and
Regts~\cite{PatelRegts} is crucial for the applications considered in
this paper. Second, our paper explicitly handles the multivariate
independence polynomial.  The work of Patel-Regts describes an
algorithm focused on the univariate case, but they mention
\cite[pp.~13]{PatelRegts} that it can be generalized to all points
dominated by $\ShearerThreshold$ (the red square region in
\Cref{fig:2}). Though it seems plausible that their method can be
extended to be applicable throughout the Shearer region (albeit still
with an exponential dependence on the slack $\alpha$), to the best of
our knowledge, the algorithmic details for doing so have not yet been
published.
\subsection{Techniques}
\SectionName{techniques}

As in previous work, our starting point is the standard
self-reducibility argument showing that the problem of designing an
FPTAS for $Z_G(\lambda)$ is equivalent to the problem of designing an
FPTAS for computing the \emph{occupation ratio} $r_v$ of a given
vertex $v$ in any given graph $G$ (i.e., the ratio of the total
weights of the independent sets containing $v$ to the total weight of
those that do not). In previous work, these occupation
ratios are actual likelihood ratios that can be translated to the
probability that the vertex $v$ is occupied, but because of complex
weights, we do not have the luxury of this interpretation.  However,
as in earlier work, we can still write formal recurrences for these
occupation ratios.  As Weitz showed~\cite{Weitz},
this recursive computation is naturally structured as a tree which has
the same structure as the the tree $\SAWTree(v, G)$ whose nodes
correspond to self-avoiding walks in $G$ starting at $v$.\footnote{
  Interpreting the computation tree as the $\SAWTree$ tree
  will have less prominence in our analysis than in previous work.}
However the tree $\SAWTree(v, G)$
has size exponential in $|V|$, so this reduction does not
immediately give a polynomial time algorithm.

The crucial step in correlation decay algorithms is to show
that this tree can be truncated to polynomial size without incurring a
large error in the value computed at the root.  In earlier work
on positive activities, especially since Restrepo et
al.~\cite{restrepo11:_improv_mixin_condit_grid_count}, the standard
method for doing this has been to consider instead a recurrence for an
appropriately chosen function $\phi(r_v)$, known as the
\emph{message}, that is chosen so that when correlation decay
holds on the $d$-regular tree, each step of the recurrence on the
truncated $\SAWTree$ contracts the error introduced by the truncation
by a constant factor.  Thus, by expanding the tree to
$\ell = O(\log\frac{n}{\epsilon})$ levels (so $d^\ell = \poly{n, \frac{1}{\epsilon}}$ nodes),
one obtains a $(1 + O(\frac{\epsilon}{n}))$-approximation to the value at the root.

Our approach also involves truncating the computation tree at an
appropriate depth and then controlling the errors introduced due to
truncation. However, in part because of the lack of a uniform bound on
the vertex activities, we are not able to recreate a message-based approach.
Instead, we perform a direct amortization argument, where we define
recursively for each node in the computation tree an \emph{error
  sensitivity parameter}, and then measure errors at that node as a
fraction of the local error sensitivity parameter.  We then establish
two facts: (1) that the errors, when measured as a fraction of the
error sensitivity parameter, do indeed decay by a constant fraction
(roughly $(1 -\Omega(\sqrt{\alpha}))$ when the input probability
vector has slack $\alpha$) at each step of the recurrence (even though
the absolute errors may not), and (2) that the error sensitivity
parameter of the root node is not too large, so that the absolute
error of the final answer can be appropriately bounded.  The detailed
argument appears in \Cref{sec:corr-decay-with-2}.
For readers familiar with the earlier
manuscript~\cite{harvey2016computing} of this paper, we point out that
in that manuscript, the decay at each step of the recurrence could
only be shown to be of the form $(1 - O(\alpha))$; informal and formal
descriptions of how this is improved to $(1 - \Omega(\sqrt{\alpha}))$
in the present paper also appear in \Cref{sec:corr-decay-with-2}.

%

%

%


\section{Overview of the correlation decay method}
\label{sec:preliminaries}
\newcommand{\CIP}{ComputeIndependencePolynomial}
\newcommand{\OcR}{OccupationRatio}
\ifbool{twocol}{
    \renewcommand{\CIP}{ComputeIndepPoly}
    \renewcommand{\OcR}{OccRatio}
}

In this section we summarize the basic concepts and facts relating to Weitz's correlation decay method.
Since all the claims are simple or known, the proofs are omitted or appear in the appendix.

\paragraph{Partition functions and occupation ratios.}
Since we are primarily interested in the hard-core partition function
(i.e., independence polynomial) with \emph{negative} activities, 
it will be convenient to introduce the following notation.
Let $G = (V, E)$ be a fixed graph,
and let $\vec{p}$ be a fixed vector of (possibly complex)
parameters on the vertices of $V$.
For $S \subseteq V$, let
$\Ind(S) = \Ind_G(S) = \setst{ I \subseteq S }{ I \text{ independent in } G }$.
Following the notation of \cite{Kolipaka,HV15},
we define the \emph{alternating-sign independence polynomial} for any subset $S$ of $V$ to be
\begin{equation*}
  \qdown_S \:=\: \qdown_S(\vec{p})
  \defeq
   \sum_{
    \substack{I \in \Ind(S)}
  } (-1)^{\abs{I}} \prod_{v \in I}p_v.
\end{equation*}
Note that $\qdown_V(\vec{p}) = Z_G(-\vec{p})$.  The computation of
$\qdown_V$ will be reduced to the computation of 
\emph{occupation ratios} defined as follows.  For a pair
$(S, u)$, where $S \subseteq V$ and $u \in S$, the occupation ratio
$r_{S,u}$ is
\begin{equation}
  \label{eq:10}
  r_{S,u} = r_{S,u}(\vec{p}) \defeq 
 - \frac{\sum_{I \in \Ind(S), u \in I} (-1)^{|I|} \prod_{v \in I} p_v}{\sum_{I \in \Ind(S), u \notin I} (-1)^{|I|} \prod_{v \in I} p_v}.
\end{equation}
For readers familiar with the notation of
Weitz~\cite{Weitz}, we note that $r_{S,u}$ agrees
with his definition of occupation ratios except for the negative signs
used in the definition here.
Using the definition of $\qdown_S$, 
and the notation $\Gamma(u) = \setst{ v }{ u \text{ is a neighbor of } v \text{ in } G}$, $\Gamma^{+}(u) = \Gamma(u) \cup \inb{u}$,
we can also rewrite this quantity as
\begin{equation}
\label{eq:11}
r_{S, u} 
  = \frac{p_u \qdown_{S \setminus \Gamma^+(u)}}{\qdown_{S\setminus\inb{u}}}
  = - \frac{\qdown_{S} - \qdown_{S\setminus\inb{u}}}{\qdown_{S\setminus \inb{u}}}
  = 1 - \frac{\qdown_S}{\qdown_{S \setminus \inb{u}}}.
\end{equation}
A standard self-reducibility argument now reduces the computation of
$\qdown_V$ to that of the $r_{S,u}$.
\begin{claim}
\label{clm:selfreduc}
  Fix an arbitrary ordering $(v_1, v_2, \dots, v_n)$ of $V$,
  and let $S_i = \inb{v_i, v_{i+1}, \dots, v_n}$. We then have
  \begin{equation}
  \label{eq:q-product}
    \qdown_V = \prod_{i=1}^{n} \frac{\qdown_{S_i}}{\qdown_{S_{i+1}}}
    = \prod_{i=1}^n(1 - r_{S_i, v_i}).
  \end{equation}
\end{claim}

\paragraph{Recurrences for the occupation ratios.}
An important observation in Weitz's
work~\cite{Weitz} was that the computation of
occupation ratios similar to the $r_{S, u}$ can be carried out over a
tree-like recursive structure. We follow a similar strategy, although
we find it convenient to work with a somewhat different notation.
\begin{definition}[\textbf{Child subproblems}]
  Given a pair $(S, u)$ with $S \subseteq V$ and $u \in S$, 
  and an arbitrary ordering $(v_1, v_2, \dots, v_k)$ of $\Gamma(u) \cap S$,
  we define the set of \emph{child} sub-problems
  $\children(S, u) = \children_{(v_1, v_2, \dots, v_k)}(S, u)$ to be
  \ifbool{twocol}{
    \begin{multline*}
      \children(S, u) ~=~ \left\{\, (S \setminus \inb{u},v_1) ,\, (S
        \setminus \inb{u,v_1},v_2) ,\right.\,\\
      \left.\dots ,\, (S \setminus \inb{u,v_1,\ldots,v_{k-1}}, v_k) \right\}
    \end{multline*}
.   }{
    \[
     \children(S, u) ~=~ \inb{\,
      (S \setminus \inb{u},v_1) ,\,
      (S \setminus \inb{u,v_1},v_2) ,\,
      \dots ,\,
      (S \setminus \inb{u,v_1,\ldots,v_{k-1}}, v_k)
      }. 
    \]
  }
\end{definition}%
\ifbool{twocol}{}{\vspace{6pt}}
Note that the ordering of neighbors used in the definition of
$\children(S, u)$ is completely arbitrary and orderings between
neighbors of different vertices do not share any consistency
constraints.  The recurrence relation for the computation of the
$r_{S, u}$, analogous to Weitz's computation tree, is then the
following:
\begin{lemma}[\textbf{Computational recurrence}]
\label{lem:recurrence}
  Fix a graph $G = (V, E)$ and a vector $\vec{p} \in \bC^V$ of complex
  parameters. Let $(S, u)$ be such that $u \in S$ and
  $S \subseteq V$.  Fix an arbitrary ordering $(v_1, v_2, \dots, v_k)$
  of $\Gamma(u) \cap S$
  and define the corresponding set
  $\children(S, u)$ of child subproblems.  We then have
  \begin{equation}
  \label{eq:tree-recursion}
    r_{S, u} = p_u \cdot \prod_{(S',u') \in \children(S,u)} \frac{1}{1 - r_{S', u'}}.
  \end{equation}
\end{lemma}
\begin{remark}
  As observed by Weitz~\cite{Weitz}, each node in the
  computation tree of $r_{V, v}$ at depth $\ell$ corresponds to a
  unique self-avoiding walk of length $\ell$ starting at $v$.
\end{remark}
\paragraph{The Shearer region.}
\label{sec:occup-rati-shear}

The Shearer region was defined (implicitly) by Shearer \cite{Shearer} as follows.

\begin{definition}[\textbf{Shearer region and slack}]
\label{def:ShearerRegion}
Given a graph $G = (V,E)$, the \emph{Shearer region} $\cS$ is the set of vectors $\vec{p} \in (0,1)^V$
such that $\qdown_S(\vec{p}) > 0$ for all $S \subseteq V$.  A
probability vector $\vec{p}$ is said to have \emph{slack} $\alpha \geq 0$
if the vector $(1 + \alpha)\vec{p}$ is also a probability vector and
is contained in $\cS$.
\end{definition}

Shearer proved that this is the maximal region of probability vectors to which the LLL
can be extended.
The equivalence with our earlier definition \eqref{eq:8} is due to Scott and Sokal
\cite[Theorem 2.10]{ScottSokal}: the Shearer region can be equivalently defined by the absence of
roots in a certain polydisc, as follows.

\begin{theorem}
\label{thm:ShearerRegionEquiv}
  A probability vector $\vec{p} \in (0,1)^V$ is in the Shearer region
  (as defined in \Cref{def:ShearerRegion})
  if and only if for all vectors
  $\vec{z} = \inp{z_v}_{v \in V}$ of complex activities such
  that $\abs{z_v} \leq p_v$, it holds that $Z_G(\vec{z}) \neq 0$.
\end{theorem}

\medskip

Considering this, it is natural to extend the definition of the Shearer region to complex parameters.
In the following, we consider primarily this complex extension of the Shearer region.

\begin{definition}[\textbf{Complex Shearer region}]
Given a graph $G = (V,E)$, the \emph{complex Shearer region} is
\[
\cS \,=\, \cS_G \defeq \setst{ \vec{p} \in \CC^V }{ Z_G(\vec{z}) \neq 0
    ~\:\forall \vec{z} \in \CC^V ,\, |z_v| \leq |p_v| }.
\]
\end{definition}
We now state some important properties of the occupation ratios and
$\qdown_S$ in the setting of real, positive parameters $\vec{p}$.
These results are essentially translations
of the results of \cite{Shearer,ScottSokal} into our notation.

\begin{lemma}[\textbf{Monotonicity and positivity of $\qdown$}]
  \label{lem:breve-prop}
  Let $G = (V, E)$ be any graph and let $\vec{p} \in (0,1)^V$ be such that $\vec{p} \in \cS$.  Then
for any subsets $A$ and $B$ of $V$ such that $A \subseteq B$, we have
    $\qdown_A(\vec{p}) \geq \qdown_B(\vec{p}) > 0$.\label{item:breve-2}
\end{lemma}

\begin{lemma}[\textbf{Occupation ratios are bounded}]
\label{lem:ratios-bounded}
  Let $G = (V, E)$ be any graph and let $\vec{p} \in (0,1)^V$ be such that $\vec{p} \in \cS$.
  Then, for any subset $S$ of $V$ and any vertex
  $u \in S$, we have $p_u \leq r_{S, u} < 1$.
\end{lemma}

\paragraph{The correlation decay algorithm.}

Weitz's high-level approach to compute the independence polynomial
is to compute the partition function via a telescoping product analogous to \eqref{eq:q-product},
and to compute each occupation ratio via a recurrence analogous to \eqref{eq:tree-recursion}.
As discussed in \Section{techniques}, the recursion is truncated
to $\ell$ levels, and the analysis shows that the occupation
ratio at the root is not affected heavily by the occupation ratios where
the truncation occurred.

\comment{
Naturally, executing this recursion without limits would result
in exponential running time. The main purpose of the technique of
{\em correlation decay} is to prove that in fact we do not have to
execute this recursion to a very large depth, because the occupation
ratio at the root is not affected heavily by the occupation ratios for
distant vertices. The main technical difficulty
of the method is to prove that this indeed happens, and to quantify
the depth that is sufficient to obtain a good estimate at the root.
}

We follow that same high-level approach here, although the details
of the analysis are quite different. 
\Algorithm{main} presents pseudocode giving a compact description of the full algorithm.
The main procedure, \textsc{\CIP}($G,\vec{p},\ell$)
implements \eqref{eq:q-product} to estimate $\breve{q}_V(\vec{p})$
for a graph $G = (V,E)$, a parameter vector $\vec{p}$,
and a desired recursion depth $\ell$. (The required value of $\ell$
depends upon the the accuracy parameter $\epsilon$; see
\Cref{thm:mainShearer}). The recursive procedure
\textsc{\OcR}($G,\vec{p},\ell,S,u$)
implements \eqref{eq:tree-recursion} to estimate the occupation ratio $r_{S,u}$ 
by executing $\ell$ levels of recursion.

\algdef{SxN}{Myprocedure}{EndMyprocedure}[2]{\algorithmicprocedure\ \textproc{#1}\ifthenelse{\equal{#2}{}}{}{(#2)}}
\algdef{SxN}{MyFor}{MyEndFor}[1]{\algorithmicfor\ #1\ \algorithmicdo}

\begin{algorithm}[t]
\caption{Our algorithm to compute a $(1+\epsilon)$-approximation to $\qdown_V(\vec{p})$.
}
\label{alg:main}
\begin{algorithmic}[1]
\Myprocedure{\CIP}{$G=(V,E),\vec{p},\ell$}
\State Fix an ordering $V = ( v_1, v_2, \ldots, v_n )$
\State $\qdown \leftarrow 1$
\MyFor{$i \leftarrow 1,\ldots,n$}
\ifbool{twocol}{
    \State $\qdown \!\leftarrow\! \qdown \cdot (1 \!-\! \text{\sc \OcR}
        (G,\vec{p},\ell,\{v_i,...,v_n\},v_i))$
}{
    \State $\qdown \leftarrow \qdown \cdot (1 - \text{\sc \OcR}
        (G,\vec{p},\ell,\{v_i,\ldots,v_n\},v_i))$
}
\MyEndFor
\State \textbf{return} $\qdown$
\EndMyprocedure

\vspace{3pt}

\Myprocedure{\OcR}{$G,\vec{p},\ell,S,u$}
\State \textbf{if} $\ell=0$ \textbf{then return} $0$
\ifbool{twocol}{
    \State Let $(w_1, \ldots, w_k)$ be a fixed ordering of $S \cap \Gamma(u)$
}{
    \State Let $(w_1, w_2, \ldots, w_k)$ denote a fixed ordering of $S \cap \Gamma(u)$
}
\State $r \leftarrow p_u$
\MyFor{$i \leftarrow 1,\ldots,k$}
    \ifbool{twocol}{
        \State $r \leftarrow \frac{ r }{ 1 \,-\, \text{\sc \OcR}
            (G,\, \vec{p},\, \ell-1,\, S \setminus \{u,w_1,\ldots,w_{i-1}\},\, w_i) }$
    }{
        \State $r \leftarrow r / (1 - \text{\sc \OcR}
            (G,\vec{p},\ell-1,S \setminus \{u,w_1,\ldots,w_{i-1}\},w_i))$
    }
\MyEndFor
\State \textbf{return} $r$
\EndMyprocedure
\end{algorithmic}
\end{algorithm}

%
%
%
%


\section{The analysis}
\label{sec:analysis}

Let us turn to the analysis of the correlation decay method in our setting.
The notion of {\em correlation decay} in the hard core model refers to the
decaying dependence of the occupation probability at a given vertex $v$
on the conditioning on a set of vertices at a certain distance from $v$.
In the setting of positive activities $\vec{z}$ \cite{Weitz},
these correlations are closely tied to the decay of errors in the computation tree for $r_{S,u}$
described in (\ref{eq:tree-recursion}). For negative or general complex activities,
the occupation ratios $r_{S,u}$ do not have a direct interpretation in terms of occupation probabilities.
However, the analysis of errors in the computation tree is reminiscent of that of
\cite{Weitz} and hence we still refer to it as {\em correlation decay}. 

Unlike Weitz's setting \cite{Weitz}, where all vertex activities are the same,
and the bounds are derived uniformly for all graphs with degrees bounded by $d$, here we are aiming
for a more refined analysis for a particular graph $G$ and a (possibly non-uniform) vector $\vec{p}$.
In Weitz's setting, the worst-case errors in the recursive tree can be proved to decay
in a uniform fashion (possibly after an application of an appropriate
potential function or message). %
That is not the case here, since the local structure of $G$ and $\vec{p}$ might cause the errors
to locally increase, even if the computation eventually converges. 
Hence it is critical to identify a local sensitivity parameter that describes how the errors
propagate in the recursive tree.

\subsection{The error sensitivity parameter}
\label{sec:corr-decay-with}

For simplicity of notation, we fix the input graph $G = (V, E)$
and an ordering on vertices $V = \{v_1,\ldots, v_n\}$ for the rest of this section.
Recall that $\vec{p} = (p_1, p_2, \dots, p_n)$ denotes a vector of vertex parameters
in the complex plane. (We have $\vec{p} = -\vec{z}$ where $\vec{z}$ are the usual activities
in the hard core model.)
We use $\abs{\vec{p}}$ to denote the vector $(\abs{p_1}, \abs{p_2}, \dots, \abs{p_n})$.
Note that $\vec{p}$ is in the Shearer region $\cS$ if and only if $\abs{\vec{p}}$ is in
$\cS$. %

First, let us consider how the errors propagate throughout the recursive computation
in Algorithm~\ref{alg:main}. Let $R_{S,u}$ an estimate obtained by the algorithm
for the occupation ratio $r_{S,u} = 1 - \breve{q}_S / \breve{q}_{S \setminus \{u\}}$.
We are interested in how the additive approximation error $|r_{S,u} - R_{S,u}|$ propagates
in the recursive computation. 

It turns out that $\vec{p}$ being real positive is in some sense the worst case; to simplify
notation, we define for $u \in S$ and $\vec{p} \in \CC^V$,
  \[
	  \ra_{S, u}(\vec{p}) \defeq r_{S, u}(\abs{\vec{p}}).
  \]
  The reader who wishes to understand the main ideas while avoiding
  some mild technical details may henceforth assume that $\vec{p}$ is
  a real positive vector, and therefore
  $\ra_{S,u}(\vec{p}) = r_{S,u}(\vec{p})$.  Indeed, an easy recursive
  argument shows that when $\vec{p} \in \cS$, $\ra_{S, u}(\vec{p})$
  dominates both $r_{S, u}(\vec{p})$ and $R_{S, u}(\vec{p})$ (see
  \Cref{clm:comp-real}):
  \begin{displaymath}
    \abs{r_{S, u}(\vec{p})} \leq \ra_{S, u}(\vec{p}) \text{ and }
    \abs{R_{S, u}(\vec{p})} \leq \ra_{S, u}(\vec{p}).
  \end{displaymath}

Now, from the mean value theorem, we obtain the following recursive bound.

\begin{claim}
\ClaimName{errorPropagate}
Let $\vec{p}$ lie in the complex Shearer region. For a node $(S,u)$ with children $\children(S,u)$ in the recursive computation tree,
we have 
\begin{equation}
    \label{eq:4}
    \abs{r_{S, u} - R_{S, u}} \leq 
    \ra_{S, u}\sum_{c \in \children(S, u)}\frac{\abs{r_c - R_c}}{1 - \ra_c}.
\end{equation}
\end{claim}

\begin{proof}
Recall Lemma~\ref{lem:recurrence}, $r_{S,u} = p_u \prod_{c \in \children(S,u)} \frac{1}{1-r_c}$.
Using the mean value theorem~(\Cref{thm:mean-value}) with $\gamma_c \defeq \ra_c$, we obtain
\ifbool{twocol}{
  \begin{align*}
    \abs{r_{S, u} - R_{S, u}}
    &\leq \abs{p_u}\prod_{c \in \children(S,
      u)}\frac{1}{1 - \ra_c}
    \sum_{c' \in \children(S, u)}\frac{\abs{r_{c'} - R_{c'}}}{1 -
      \ra_{c'}}\\
    &= \ra_{S, u}\sum_{c \in \children(S, u)}\frac{\abs{r_c - R_c}}{1 - \ra_c}.
  \end{align*}
}{
  $$
    \abs{r_{S, u} - R_{S, u}} \leq \abs{p_u}\prod_{c \in \children(S,
      u)}\frac{1}{1 - \ra_c}
    \sum_{c' \in \children(S, u)}\frac{\abs{r_{c'} - R_{c'}}}{1 - \ra_{c'}} =
    \ra_{S, u}\sum_{c \in \children(S, u)}\frac{\abs{r_c - R_c}}{1 - \ra_c}.
    $$
  }
  Note that the conditions imposed in the hypothesis of
    \Cref{thm:mean-value} hold, since, as pointed out above,
    $\abs{r_{S, u}}, \abs{R_{S, u}} \leq \ra_{S, u}$ for all nodes
    $(S, u)$ in the computation tree (see \Cref{clm:comp-real}).
\end{proof}

Our error sensitivity parameter is defined to capture how errors propagate under this recursive bound.
An important observation is that the derivative of $\ra_{S,u}((1+t) \vec{p})$ with respect to $t$ satisfies
a recurrence very similar to \Claim{errorPropagate}, and this is the main motivation behind the 
following definition.

\begin{definition}[\textbf{Error sensitivity parameter}]
\label{def:beta}
  The \emph{error sensitivity parameter}
  $\beta_{S, u} \defeq \beta_{S, u}(\vec{p})$ is defined as
  \begin{equation}
  	\label{eq:beta}
    \beta_{S, u}(\vec{p}) \defeq \left.\diff{\ra_{S,u}((1+t)\vec{p})}{t}\right\vert_{t = 0}.
  \end{equation}
\end{definition}

\begin{claim}
  \ClaimName{bRecurrence} Let $(S, u)$ be node in the computation tree. Then
  \begin{displaymath}
    \beta_{S,u}
    = \ra_{S,u} \cdot \inp{1 + \sum_{c \in \children(S,u)} \frac{\beta_c}{1 - \ra_c}}.
  \end{displaymath}
\end{claim}

\begin{proof} By a direct calculation using the definition of $\beta_{S, u}$,
  \ifbool{twocol}{
    \begin{align*}
    &\beta_{S,u}
    \defeq \diff{\ra_{S,u}((1+t)\vec{p})}{t} \Big|_{t=0}\\
    &= \diff{}{t}
      \Bigg(
      (1+t) \abs{p_u} \cdot
      \prod_{c \in \children(S,u)}
      \frac{1}{1 - \ra_c((1+t)\vec{p})}
      \Bigg)
      \Bigg|_{t=0} \\
    &\:= \prod_{c \in \children(S,u)} \frac{1}{1 - \ra_c(\vec{p})} \cdot
      \Big(\abs{p_u} + \abs{p_u} \sum_{c \in \children(S,u)}
      \frac{\beta_c}{1 - \ra_c} \Big) \\
    &\:= \ra_{S,u}(\vec{p}) \cdot \Big(
      1 + \sum_{c \in \children(S,u)}
      \frac{\beta_c}{1 - \ra_c}
      \Big). \qedhere
  \end{align*}
  }{
  \begin{align*}
    &\beta_{S,u}
    \defeq \diff{\ra_{S,u}((1+t)\vec{p})}{t} \Big|_{t=0}
      = \diff{}{t}
      \Bigg(
      (1+t) \abs{p_u} \cdot
      \prod_{c \in \children(S,u)}
      \frac{1}{1 - \ra_c((1+t)\vec{p})}
      \Bigg)
      \Bigg|_{t=0} \\
    &\:=
      \prod_{c \in \children(S,u)} \frac{1}{1 - \ra_c(\vec{p})}  \cdot
      \Bigg( \abs{p_u}
      + \Big( (1+t)\abs{p_u} \Big) \cdot \!\!
      \sum_{c \in \children(S,u)} \frac{1}{1 - \ra_c((1+t)\vec{p})} \cdot
      \diff{\ra_c((1+t)\vec{p})}{t} \Bigg)\Bigg|_{t=0} \\
    &\:= \prod_{c \in \children(S,u)} \frac{1}{1 - \ra_c(\vec{p})} \cdot
      \Big(\abs{p_u} + \abs{p_u} \sum_{c \in \children(S,u)}
      \frac{\beta_c}{1 - \ra_c} \Big) \\
    &\:= \ra_{S,u}(\vec{p}) \cdot \Big(
      1 + \sum_{c \in \children(S,u)}
      \frac{\beta_c}{1 - \ra_c}
      \Big). \qedhere
  \end{align*}
}
\end{proof}

We will now prove several additional properties of the error sensitivity parameter.

\begin{lemma}
\label{lem:r-convex}
Fix a parameter vector $\vec{p} \in \cS$. Let $t_0$ be such that for
$0 \leq t \leq t_{0}$, $(1 + t)\vec{p}$ is also in $\cS$.
Define $\beta_{S, u}(\vec{p}, t) = \diff{}{t}\ra_{S, u}((1 + t)\vec{p})$. 
(Note that Definition~\ref{def:beta} is consistent with
 $\beta_{S, u}(\vec{p}) = \beta_{S,u}(\vec{p},0)$.)
Then,
for all nodes $(S, u)$ in the computation tree, $\beta_{S, u}(\vec{p}, t)$
is a non-negative, non-decreasing function of $t$ for $t \in [0, t_0]$.
Thus, the map $t \mapsto \ra_{S, u}((1+t)\vec{p})$ is non-decreasing and convex
over the same domain.
\end{lemma}

\begin{proof}
  We induct on $|S|$.
  The base case is $S=\inb{u}$, so $\ra_{\{u\},u}((1+t)\vec{p}) = (1+t)\abs{p_u}$.
  We therefore have
  \begin{displaymath}
    \beta_{\{u\}, u}(\vec{p}, t) = \diff{}{t} \ra_{\{u\},u}((1+t)\vec{p}) = |p_u|,
  \end{displaymath}
  which is a constant (and hence non-decreasing), non-negative function of $t$.
  
  For the inductive case, we use a recursive formula for
  $\beta_{S, u}(p,t)$ as in the proof of \Cref{clm:bRecurrence}. We have
  \ifbool{twocol}{
    \begin{multline*}
      \beta_{S, u}(\vec{p}, t) = \abs{p_u} \cdot \hspace{-3ex}
        \prod_{(S',u') \in \children(S, u)} \frac{1}{1 -
          \ra_{S',u'}((1 + t)\vec{p})} \\
        + \ra_{S, u}((1+t)\vec{p})
        \cdot \hspace{-3ex} \sum_{(S',u') \in \children(S,u)}
        \frac{\beta_{S',u'}(\vec{p}, t)}{1 -
          \ra_{S',u'}((1+t)\vec{p})}.
    \end{multline*}
  }{
    \[
      \beta_{S, u}(\vec{p}, t)
      = {\abs{p_u} \cdot \hspace{-3ex} \prod_{(S',u') \in \children(S, u)}
        \frac{1}{1 - \ra_{S',u'}((1 + t)\vec{p})} \:+\: \ra_{S, u}((1+t)\vec{p})
        \cdot \hspace{-3ex} \sum_{(S',u') \in \children(S,u)}
        \frac{\beta_{S',u'}(\vec{p}, t)}{1 - \ra_{S',u'}((1+t)\vec{p})} }.
    \]
  }
  \noindent
  By the induction hypothesis, $\beta_{S',u'}(\vec{p}, t) \geq 0$ for each
  $|S'| < |S|$, $u' \in S'$. Since $(1+t)\vec{p} \in \cS$, \Lemma{ratios-bounded} implies
  $0 \leq \ra_{S',u'}((1+t)\vec{p}) < 1$.
  Therefore $\beta_{S, u}(\vec{p}, t) \geq 0$ as well.
  Moreover, the inductive hypothesis
  implies that both $\ra_{S',u'}((1+t)\vec{p})$ and $\beta_{S', u'}(\vec{p}, t)$ are
  non-decreasing in $t$.  Since the whole expression is monotone in
  $\ra_{S',u'}$ and $\beta_{S', u'}(\vec{p}, t)$, the left-hand side
  $\beta_{S, u}(\vec{p}, t)$ is also a non-decreasing function of $t$.
\end{proof}

We can now prove the following relations between the $\beta_{S,u}$
and the $\ra_{S,u}$.

\begin{lemma}
  \label{lem:beta-prop}
  Let $\vec{p} \in \bC^V$ and $\alpha > 0$ satisfy
  $(1 + \alpha)\vec{p} \in \cS$.  We then have the following
  inequalities for all nodes $(S, u)$ in the computation tree. (We use
  the shorthand notation $\ra_{S, u} = \ra_{S, u}(\vec{p})$ and
  $\beta_{S, u} = \beta_{S, u}(\vec{p})$.)
  \begin{enumerate}
  \item $\beta_{S, u} < (1 - \ra_{S, u})/\alpha$. \label{item:1}
  \item $\ra_{S, u} \leq \beta_{S, u} \leq (1+d_u/\alpha)\cdot \ra_{S,
      u}$, where $d_u$ is the degree of the vertex $u$ in $G$. \label{item:2}
  \item
    $\ra_{S, u} \le \frac{1}{1 + \alpha}\ra_{S, u}((1 + \alpha)\vec{p}) <
    \frac{1}{1 + \alpha}$. \label{item:3}
  \end{enumerate}
\end{lemma}

\begin{proof}
  Since $(1 + \alpha)\vec{p} \in \cS$, \Lemma{ratios-bounded} implies $\ra_{S,
    u}((1+\alpha)\vec{p}) < 1$.  Further, from \Cref{lem:r-convex}, we know
  that $t \mapsto \ra_{S,u}((1+t)\vec{p})$ is convex for $t \in [0,
  \alpha]$. 
  \Cref{item:1} of the lemma then follows from the inequalities
  \begin{equation}
    1 ~>~ \ra_{S,u}((1+\alpha) \vec{p})
    ~\geq~ \ra_{S,u} + \alpha \diff{\ra_{S,u}}{t} \Big|_{t=0}
    ~\geq~ \ra_{S,u} + \alpha \beta_{S,u}. \label{eq:3}
  \end{equation}

\Cref{item:2} follows from \Cref{clm:bRecurrence}, using $\beta_{S,u}
\geq 0$, $\abs{\children(S, u)} \leq d_u$, and $\beta_{S,u} < (1-\ra_{S,u}) / \alpha$ from \cref{item:1}.

To prove the first inequality in \cref{item:3}, we again use \cref{eq:3},
and substitute $\beta_{S,u} \geq \ra_{S,u}$ from \cref{item:2}. The
second inequality follows from the fact that $\ra_{S, u}((1+\alpha)\vec{p}) < 1$
as mentioned above. \qedhere
\end{proof}

\vspace*{-4pt} 

Finally we relate the quantities $\ra_{S, u}$ to the
quantities $r_{S, u}$ that we actually want to approximate.

\begin{claim}
  \label{clm:comp-real}
  Let $\vec{p}$ lie in the complex Shearer region.
  For any node $(S, u)$ in the computation tree, we have
  \begin{displaymath}
    \abs{r_{S, u}(\vec{p})} \leq \ra_{S, u}(\vec{p}) \text{ and }
    \abs{R_{S, u}(\vec{p})} \leq \ra_{S, u}(\vec{p}).
\vspace*{-4pt} 
  \end{displaymath}
\end{claim}

\begin{proof}
  For both $r_{S, u}$ and $R_{S, u}$, the proof is by induction on
  $|S|$.  The base case for $r_{S, u}$ is when $S$ is a singleton, in
  which case we have
  $\abs{r_{\inb{u}, u}(\vec{p})} = \abs{p_u} = \ra_{\inb{u},
    u}(\vec{p})$.  For $R_{S, u}$, the base case is when $(S, u)$ is at
  depth $\ell$ in the computation tree (where $\ell$ is as in the
  input to \Cref{alg:main}), in which case one has
  $R_{S, u}(\vec{p}) = 0 \leq \rho_{S, u}(\vec{p})$.  For the inductive case, we use the
  recursion for $r_{S, u}(\vec{p})$ to obtain
  \ifbool{twocol}{
    \begin{align*}
    \abs{r_{S, u}(\vec{p})} &= \abs{p_u\prod_{c \in \children(S, u)}\frac{1}{1 -
        r_c(\vec{p})}}\\
    &\leq \abs{p_u}\prod_{c \in \children(S, u)} \frac{1}{1 -
      \abs{r_c(\vec{p})}}\\
    &\leq \abs{p_u}\prod_{c \in \children(S, u)} \frac{1}{1 -
      \ra_c(\vec{p})} = \ra_{S, u}.
  \end{align*}
  }{
  \begin{displaymath}
    \abs{r_{S, u}(\vec{p})} = \abs{p_u\prod_{c \in \children(S, u)}\frac{1}{1 -
        r_c(\vec{p})}} \leq \abs{p_u}\prod_{c \in \children(S, u)} \frac{1}{1 -
      \abs{r_c(\vec{p})}} \leq \abs{p_u}\prod_{c \in \children(S, u)} \frac{1}{1 -
      \ra_c(\vec{p})} = \ra_{S, u}.
  \end{displaymath}
  }
  Here, the first inequality follows from \Cref{fct:simple} since the
  induction hypothesis implies that $\vec{p} \in \cS$
  $\abs{r_c(\vec{p})} \leq \ra_c(\vec{p})$, while that fact that
  $\vec{p} \in \cS$ implies that $\ra_c(\vec{p}) < 1$ (e.g., from
  \cref{item:3} of \Cref{lem:beta-prop}). The second inequality
  follows directly from the induction hypothesis.  The inductive step
  for $R_{S, u}$ is identical.
\end{proof}

%
%
%
%


\subsection{Correlation decay with complex activities}
\label{sec:corr-decay-with-2}

We now use the error sensitivity parameters
to establish the correlation decay results
needed for our FPTAS.

  Let $G = (V, E)$ be a graph on $n$ vertices, and let $\vec{p} \in \CC^V$ be
  such that $(1+\alpha)^2 \vec{p} \in \cS$ ($\vec{p}$ is in the Shearer region with slack $\simeq 2 \alpha$).
  The root of the recursion is a pair $(A,a)$ where $A \subseteq V$ and $a \in A$.
  Let $\ell \geq 0$ be arbitrary.
  Recall that Algorithm~\ref{alg:main} recursively computes an estimate $R_{S,u}$ of $r_{S,u}$,
  where for every pair $(S,u)$ encountered, we have
  \ifbool{twocol}{$R_{S,u} = 0$ if $(S, u)$ is at depth $\ell$ in the computation
        tree, and $R_{S, u} = p_u \cdot \prod_{c \in \children(S,u)}
        (1 - R_c)^{-1}$ otherwise. }{
  \begin{displaymath}
    R_{S,u} ~=~
    \begin{cases}
      0 &\quad\text{(if $(S, u)$ is at depth $\ell$ in the computation
        tree)} \\
      p_u \cdot \prod_{c \in \children(S,u)} (1 - R_c)^{-1}
      &\quad\text{(otherwise)}
    \end{cases}.
  \end{displaymath}
  }
  The depth $\delta$ of a node is defined as its distance from the root $(A,a)$ in the recursive tree:
  the root has $\delta(A, a)=0$,
  its children have $\delta(S, u)=1$, etc.
  
  \paragraph{Intuition.}
  \Cref{clm:selfreduc} implies that it is sufficient to get good
  approximations for the occupation ratios in order to obtain an
  FPTAS.  Suppose now that we were to expand the computation tree for
  computing a particular occupation ratio up to depth $\ell$ as
  described above, and were then able to show that at every node
  $(S, u)$ in this tree, the approximation error $|r_{S,u} - R_{S,u}|$
  is smaller than the maximum approximation error at the node's
  children by a factor $c < 1$.  It would then follow that the
  approximation error at the root node is $O(c^\ell)$, and hence that
  it is sufficient to take $\ell = O(\log n)$ in order to obtain an
  inverse polynomial approximation of the occupation ratio (which in
  turn can be shown to be sufficient for obtaining an FPTAS for $Z$).
  However, since degrees and the activity parameters might vary
  throughout the tree, the errors $|r_{S,u} - R_{S,u}|$ do not decay
  uniformly in this fashion at each node of the tree; they might even
  increase locally. (Examples are not difficult to construct.)
  Instead, we aim to use the error sensitivity parameter $\beta_{S,u}$
  as a yardstick against which the approximation error
  $|r_{S,u} - R_{S,u}|$ at node $(S, u)$ ought to be compared.

  As we mentioned earlier, $\beta_{S,u}$ is a natural error
  sensitivity parameter because it satisfies a recurrence
  (\Claim{bRecurrence}) similar to the recurrence for error
  propagation (\Claim{errorPropagate}).  In an earlier version of this
  paper, we used ``normalized errors'' roughly of the form
  $|r_{S,u} - R_{S,u}| / \beta_{S,u}$, and showed that they decay by a
  factor of $1-\Theta(\alpha)$ at each level of the tree.  Here we
  present an improved analysis which leads to a decay factor of
  $1-\Theta(\sqrt{\alpha})$, which is in fact tight (in particular on
  the infinite $d$-regular tree, see \Cref{sec:optim-decay-rate}).
  
  The improvement comes from analyzing in conjunction the behavior of
  $r_{S,u}$ at two different probability vectors: $\vec{p}$ and the
  vector $(1+ \alpha) \vec{p}$ of slightly larger probabilities (which
  nonetheless still has a slack of $\alpha$).  We denote
  $\beta'_{S,u} =
  \frac{d\ra_{S,u}((1+t)\vec{p})}{dt}\Big|_{t=\alpha}$.  Instead of
  $\beta_{S,u}$, we compare the errors to the quantity
  $\sqrt{\beta_{S,u} \beta'_{S,u}}$, i.e., we normalize the error at
  node $(S, u)$ as
  $|r_{S,u} - R_{S,u}| / \sqrt{\beta_{S,u}\beta'_{S,u}}$.

  The reason for this choice of the normalization is as follows.  It
  can be shown that the case where the earlier version of the
  normalized error decays by a factor of only $1 - O(\alpha)$
  corresponds to the situation where
  $\ra_{S,u}((1+\alpha) \vec{p}) \simeq (1+\alpha)
  \ra_{S,u}(\vec{p})$. But in that case, an argument based on
  \Cref{lem:r-convex,lem:beta-prop} implies that
  $\beta_{S,u} - \ra_{S, u}$ must be quite small.  The new
  normalization of the error allows us to exploit this phenomenon: we
  can now show, roughly speaking, that the smaller of these two
  factors, i.e., $\ra_{S,u}(\vec{p})/\ra_{S,u}((1+\alpha) \vec{p})$
  (which was the factor obtained in the analysis of the earlier
  version) on the one hand, and $\beta_{S,u} - \ra_{S, u}$ on the
  other, can be taken to be the decay factor for the new normalized
  error.  We then show that at least one of these two factors is as
  small as $1 - \Omega(\sqrt{\alpha})$.  At a technical level, proving
  this requires a careful comparison of the recurrences for the
  propagation of unnormalized errors (\Claim{errorPropagate}) with the
  recurrence for the $\beta_{S, u}$ (\Claim{bRecurrence}), exploiting
  in particular the extra additive term of $1$ in the latter
  recurrence.  We now make this intuition precise in the following
  theorem. Recall that $\vec{p}$ is assumed to be such that
  $(1+\alpha)^2\vec{p} \in \cS$.

  \begin{theorem}
  \TheoremName{decayShearer}
  For notational simplicity, let $\ra_{S,v}=\ra_{S, v}(\vec{p})$ and
  $\raPrime{S,v}=\ra_{S,v}((1+\alpha)\vec{p})$.
  Similarly, let $\beta_{S,v} = \beta_{S, v}(\vec{p})$ and
  $\beta'_{S,v} = \beta_{S,v}((1+\alpha) \vec{p})$.
  For a node $(S,u)$ in a computation tree of depth $\ell$,
  \[
  \abs{r_{S,u} - R_{S, u}} \leq
  \sqrt{\beta_{S,u}\beta'_{S,u}}{(1+\sqrt{\alpha})^{-(\ell-\depth(S, u))/2}}.
  \]
\end{theorem}

\begin{proof}
  The proof is by induction on $\ell - \depth(S, u)$.  The base case
  is $\depth(S, u) = \ell$ and $R_{S, u} = 0$; we want to prove $|r_{S,u}| \leq \sqrt{\beta_{S,u} \beta'_{S,u}}$.
  We have $\abs{r_{S, u}} \leq \ra_{S,u}$ (\Cref{clm:comp-real}).
  The base case  follows since $\ra_{S, u} \leq \beta_{S, u}$ and
  $\ra_{S, u} \leq \raPrime{S, u} \leq \beta'_{S,u}$, by \Cref{lem:r-convex}, and \cref{item:2} of
  \Cref{lem:beta-prop}.

  For the inductive step, we apply the recursive formula from \Claim{errorPropagate}:
  $$
    \abs{r_{S, u} - R_{S, u}} \leq 
    \ra_{S, u}\sum_{c \in \children(S, u)}\frac{\abs{r_c - R_c}}{1 - \ra_c}.
  $$
  By definition, $\depth(c) = \depth(S, u) + 1$ for all $c \in \children(S, u)$.
  By the induction hypothesis, we therefore have
  \ifbool{twocol}{
    \begin{align}
    \abs{r_{S, u} - R_{S, u}}
    &\leq
      \ra_{S, u}\sum_{c \in \children(S,u)} \frac{\sqrt{\beta_c\beta'_c}}{1 - \ra_c}
         (1+\sqrt{\alpha})^{-(\ell - \depth(S,u) - 1)/2} \nonumber\\
    &\leq
      \frac{\ra_{S, u} \sum_{c \in \children(S, u)}
      \sqrt{\frac{\beta_c}{1 - \ra_c}}
      \cdot \sqrt{\frac{\beta'_c}{1 - \raPrime{c}}}}
      {(1+\sqrt{\alpha})^{(\ell - \depth(S,u) - 1)/2}},\nonumber\\
    &\leq
      \frac{\ra_{S, u}  \sqrt{\sum_{c \in \children(S, u)}
      \frac{\beta_c}{1 - \ra_c}}
      \cdot
      \sqrt{
      \sum_{c \in \children(S, u)}
      \frac{\beta_c'}{1 - \raPrime{c}}
      }}{(1+\sqrt{\alpha})^{(\ell - \depth(S,u) - 1)/2}} \nonumber\\
    &= \frac{\sqrt{\frac{\ra_{S, u}}{\raPrime{S, u}}}
      \sqrt{\beta_{S, u}  - \ra_{S, u}}
      \cdot
      \sqrt{\beta_{S,u}' - \raPrime{S, u}}}{
      (1+\sqrt{\alpha})^{(\ell - \depth(S,u) - 1)/2}}\label{eq:5}
  \end{align}
  }{
    \begin{align}
    \abs{r_{S, u} - R_{S, u}}
    &\leq
      \ra_{S, u}\sum_{c \in \children(S,u)} \frac{\sqrt{\beta_c\beta'_c}}{1 - \ra_c}
         (1+\sqrt{\alpha})^{-(\ell - \depth(S,u) - 1)/2} \nonumber\\
    &\leq
      \ra_{S, u} \sum_{c \in \children(S, u)}
      \sqrt{\frac{\beta_c}{1 - \ra_c}}
      \cdot \sqrt{\frac{\beta'_c}{1 - \raPrime{c}}}
		\ (1+\sqrt{\alpha})^{-(\ell - \depth(S,u) - 1)/2},\nonumber\\
    &\leq
      \ra_{S, u}  \sqrt{\sum_{c \in \children(S, u)}
      \frac{\beta_c}{1 - \ra_c}}
      \cdot
      \sqrt{
      \sum_{c \in \children(S, u)}
      \frac{\beta_c'}{1 - \raPrime{c}}
      } \ (1+\sqrt{\alpha})^{-(\ell - \depth(S,u) - 1)/2} \nonumber\\
    &= \sqrt{\frac{\ra_{S, u}}{\raPrime{S, u}}}
      \sqrt{\beta_{S, u}  - \ra_{S, u}}
      \cdot
      \sqrt{\beta_{S,u}' - \raPrime{S, u}}
      (1+\sqrt{\alpha})^{-(\ell - \depth(S,u) - 1)/2} \label{eq:5}
    \end{align}
  }
  where the second inequality uses $\ra_c \leq \raPrime{c}$ (which
  follows from \cref{item:3} of \Cref{lem:beta-prop}), the third is
  the Cauchy-Schwarz inequality, and the last equality uses the
  recursion for $\beta_{S,u}$ as developed in
  \Cref{clm:bRecurrence}.\footnote{We implictly assume here and later
    in this proof that $\raPrime{S, u}$ is strictly positive.  For, if
    $\raPrime{S, u}$ were $0$, then by \cref{item:3} of
    \Cref{lem:beta-prop}, $\ra_{S, u}$ and $R_{S, u}$ will also be
    $0$ and the inductive hypothesis will be trivially true.} Note
  that it also follows from \cref{item:2,item:3} of
  \Cref{lem:beta-prop} that
  \begin{equation}
    \label{eq:6}
    0 \leq \ra_{S, u} \leq \raPrime{S, u} \text{ and }
    \beta_{S, u} - \ra_{S, u}, \beta'_{S, u} - \raPrime{S, u} \ge 0.
  \end{equation}
  We now divide the rest of the analysis into two cases.
  \begin{description}
  \item[Case 1: $\beta_{S, u} \leq \ra_{S, u}/\sqrt{\alpha}$].  In
    this case, we have
    $ 0\leq \beta_{S, u} - \ra_{S, u} \leq \beta_{S, u}(1 -
    \sqrt{\alpha}) \leq \beta_{S, u}/(1 + \sqrt{\alpha})$.
    Substituting this into \cref{eq:5}, and estimating the rest of the
    factors using \eqref{eq:6}, we get
    \ifbool{twocol}{
      \begin{align*}
      \abs{r_{S, u} - R_{S, u}}
      &\leq 
        \frac{\sqrt{\frac{\ra_{S, u}}{\raPrime{S, u}}}
        \sqrt{\beta_{S, u}  - \ra_{S, u}}
        \cdot
        \sqrt{\beta_{S,u}' - \raPrime{S, u}}}{
        (1+\sqrt{\alpha})^{(\ell - \depth(S,u) - 1)/2}}\\
      &\leq 
        \frac{\sqrt{\beta_{S, u}\beta_{S,u}'}}{(1+\sqrt{\alpha})^{(\ell - \depth(S,u) - 1)/2}}\cdot \frac{1}{\sqrt{1 + \sqrt{\alpha}}},
    \end{align*}
    }{\begin{align*}
      \abs{r_{S, u} - R_{S, u}}
      &\leq 
        \sqrt{\frac{\ra_{S, u}}{\raPrime{S, u}}}
        \sqrt{\beta_{S, u}  - \ra_{S, u}}
        \cdot
        \sqrt{\beta_{S,u}' - \raPrime{S, u}}
        \ (1+\sqrt{\alpha})^{-(\ell - \depth(S,u) - 1)/2}\\
      &\leq 
        \sqrt{\beta_{S, u}\beta_{S,u}'}\cdot \frac{1}{\sqrt{1 + \sqrt{\alpha}}}
        \ (1+\sqrt{\alpha})^{-(\ell - \depth(S,u) - 1)/2},
      \end{align*}
      }
    which completes the inductive step in this case.
  \item[Case 2: $\beta_{S, u} > \ra_{S, u}/\sqrt{\alpha}$].  We
    claim that in this case,
    $\ra_{S, u} \leq \raPrime{S, u}/(1 + \sqrt{\alpha})$.  Indeed,
    using the same argument as in the proof of \cref{item:1} of
    \Cref{clm:bRecurrence}, we have
    \ifbool{twocol}{
      \begin{multline}
      \raPrime{S, u} ~\defeq~ \ra_{S,u}((1+\alpha) \vec{p})
      ~\geq~ \ra_{S,u} + \alpha \diff{\ra_{S,u}}{t} \Big|_{t=0}\\
      ~\geq~ \ra_{S,u} + \alpha \beta_{S,u} ~>~ (1 + \sqrt{\alpha}) \ra_{S, u}.
      \label{eq:12}
    \end{multline}
    }{\begin{equation}
      \raPrime{S, u} ~\defeq~ \ra_{S,u}((1+\alpha) \vec{p})
      ~\geq~ \ra_{S,u} + \alpha \diff{\ra_{S,u}}{t} \Big|_{t=0}
      ~\geq~ \ra_{S,u} + \alpha \beta_{S,u} ~>~ (1 + \sqrt{\alpha}) \ra_{S, u}.
      \label{eq:12}
    \end{equation}}
    Substituting this into \cref{eq:5}, we obtain
    \ifbool{twocol}{\begin{align*}
      \abs{r_{S, u} - R_{S, u}}
      &\leq 
        \frac{\sqrt{\frac{\ra_{S, u}}{\raPrime{S, u}}}
        \sqrt{\beta_{S, u}  - \ra_{S, u}}
        \cdot
        \sqrt{\beta_{S,u}' - \raPrime{S, u}}}{
        (1+\sqrt{\alpha})^{(\ell - \depth(S,u) - 1)/2}}\\
      &\leq 
        \frac{1}{\sqrt{1 + \sqrt{\alpha}}}\cdot
        \frac{\sqrt{\beta_{S, u}\beta_{S,u}'}}{
        (1+\sqrt{\alpha})^{(\ell - \depth(S,u) - 1)/2}},
    \end{align*}
  }{
      \begin{align*}
      \abs{r_{S, u} - R_{S, u}}
      &\leq 
        \sqrt{\frac{\ra_{S, u}}{\raPrime{S, u}}}
        \sqrt{\beta_{S, u}  - \ra_{S, u}}
        \cdot
        \sqrt{\beta_{S,u}' - \raPrime{S, u}}
        \ (1+\sqrt{\alpha})^{-(\ell - \depth(S,u) - 1)/2}\\
      &\leq 
        \frac{1}{\sqrt{1 + \sqrt{\alpha}}}
        \sqrt{\beta_{S, u}\beta_{S,u}'}
        \ (1+\sqrt{\alpha})^{-(\ell - \depth(S,u) - 1)/2},
      \end{align*}
    }
    which establishes the induction step in this case as well.\qedhere
\end{description}
\end{proof}

\begin{corollary}
  \CorollaryName{decayRoot} Given a graph $G = (V, E)$, let
  $\vec{p}$ be a complex parameter vector such that $(1+\alpha)^2
  \vec{p} \in \cS$.  Let $(A,
  a)$ be the root of the recursive computation of
  \Theorem{decayShearer}, where $a$ is a vertex of degree $d_a$ in
  $G$. Then, we have
\[
  \abs{r_{A, a} - R_{A, a}} \leq \frac{1 + d_a/\alpha}{(1 + \sqrt{\alpha})^{\ell/2}}.
\]
\end{corollary}
\begin{proof}
  Recall from \Cref{lem:ratios-bounded} that $\ra_{A, a} \in [0, 1)$ when $\vec{p} \in \cS$.
  Then \cref{item:2} of \Cref{lem:beta-prop} implies that
\[
\beta_{A,a}\big( (1+\alpha) \vec{p} \big)
    ~\leq~ (1+d_a/\alpha) \cdot r_{A,a}\big( (1+\alpha) \vec{p} \big)
    ~\leq~ 1+d_a/\alpha.
\]
We can apply \Cref{lem:beta-prop} to $\beta_{A,a}((1+\alpha)
\vec{p})$ (as opposed to $\beta_{A,
  a}(\vec{p})$) because of the assumption that $(1+\alpha)^2 \vec{p}
\in S$. The claim now follows from \Theorem{decayShearer} since $(A,
a)$ is at depth $0$ in the computation tree.
\end{proof}

\medskip

We can now prove that our algorithm indeed provides an FPTAS for the quantity $\qdown_V(\vec{p})$
(for bounded degree graphs and constant slack).
We remark that this also proves \Theorem{mainShearerIntro}.

\begin{theorem}[\textbf{FPTAS for $\qdown$}]
  \TheoremName{mainShearer} Given $\alpha, \epsilon \in (0,1]$, a
  graph $G = (V, E)$ on $n$ vertices with maximum degree $d$, and a
  parameter vector $\vec{p}$ such that $(1+\alpha)^2 \vec{p} \in \cS$,
  a $(1+\epsilon)$-approximation to $\qdown_V(\vec{p})$ can be
  computed in time
  $(\frac{n}{\epsilon \alpha})^{O(\log (d)/ \sqrt{\alpha})}$.
\end{theorem}

\begin{proof}
  Order the vertices of $G$ arbitrarily as $v_1, v_2, \dots, v_n$. Recall that
  \[
    \qdown \defeq \qdown_V = \prod_{i = 1}^n\inp{1 - r_{S_i, v_i}},
  \]
  where $S_i \defeq \inb{v_i, v_{i+1}, \dots, v_n}$.  Let us denote by
  capital letters the estimates computed by Algorithm~\ref{alg:main}.
  $R_{S_i, v_i}$ is computed using $\ell$ levels of the recurrence in
  \Cref{thm:decayShearer}, where 
    $$\ell = \ceil{ 2\log_{(1 + \sqrt{\alpha})}\inp{ \frac{2 (1 +
        \alpha)(1 + d/\alpha) n}{\epsilon\alpha}}}.$$ 
  We have
  \ifbool{twocol}{
    \begin{multline*}
      \ell = \left\lceil \frac{2}{\log (1+\sqrt{\alpha})} \log \inp{
          \frac{2 (1 + \alpha)(1 + d/\alpha) n}{\epsilon\alpha}}
      \right\rceil\\
      \leq O\inp{\frac{1}{\sqrt{\alpha}} \log \inp{
          \frac{n}{\epsilon \alpha}}}.
    \end{multline*}
}{
    $$ \ell = \left\lceil \frac{2}{\log (1+\sqrt{\alpha})} \log \inp{ \frac{2 (1 +
        \alpha)(1 + d/\alpha) n}{\epsilon\alpha}} \right\rceil 
  \leq O\inp{\frac{1}{\sqrt{\alpha}} \log \inp{ \frac{n}{\epsilon \alpha}}}. $$}
  The number of nodes of the computation
  tree explored in the computation of each $R_{S_i, v_i}$ is
  $O(d^\ell)$ since the graph is assumed to be of degree at most $d$.
  This proves the running time bound.
    
  The algorithm outputs the quantity
  \begin{displaymath}
    \breve{Q} \defeq \prod_{i=1}^n(1 - R_{S_i, v_i}).
  \end{displaymath}
  We now prove that this is indeed a $(1+\epsilon)$-approximation
  for $\qdown_V(\vec{p})$.  For ease of notation, we define
  $\xi_i \defeq 1 - r_{S_i, v_i}$ and $\Xi_i = 1 - R_{S_i, v_i}$.
  From \Cref{cor:decayRoot} we have, for each $i$,
  \begin{equation}
    \abs{\frac{\Xi_i}{\xi_i} - 1}
    = \frac{\abs{r_{S_i, v_i} - R_{S_i, v_i}}}{\abs{1 - r_{S_i, v_i}}}
    \leq
    \frac{1 + \alpha}{\alpha}
    \cdot \frac{1 + d/\alpha}{(1+\sqrt{\alpha})^{\ell/2}}. \label{eq:2}
  \end{equation}
  Here, the last inequality uses \Cref{clm:comp-real} (which implies
  that $\abs{r_{S_i, v_i}} \leq \ra_{S_i, v_i}$) and \cref{item:3}
  from \Cref{lem:beta-prop} (which implies that when $\vec{p} \in \cS$,
  $\ra_{S_i, v_i} \leq \frac{1}{1 + \alpha} < 1$).  Together, with
  \Cref{fct:simple}, these two inequalities imply that
  $\frac{1}{\abs{1 - r_{S_i, v_i}}} \leq \frac{1+\alpha}{\alpha}$.
  Since
  $\ell = \ceil{
    2\log_{(1 + \sqrt{\alpha})}\inp{
      \frac{2 (1 + \alpha)(1 + d/\alpha) n}{\epsilon\alpha}}}$,
  we therefore have
    $\abs{\frac{\Xi_i}{\xi_i} - 1} \leq \frac{\epsilon}{2n}$, for all $i$.
  Combining this with \Cref{fct:multiply},
  and recalling that $\qdown = \prod_{i=1}^n\xi_i$ and that $\breve{Q} = \prod_{i=1}^n \Xi_i$,
  we obtain $\abs{\breve{Q} - \qdown} \leq \epsilon\abs{\qdown}$, which proves the theorem.
\end{proof}

%
%
%
%


\section{Application for the Lov\'{a}sz Local Lemma}
\label{sec:ShearerVar}
\subsection{Proof of Shearer's lemma by rounding variables}
\SectionName{existential}

Let us recall Shearer's formulation~\cite{Shearer} of the Lov\'{a}sz Local Lemma,
as stated in \Cref{app:LLL}.
For any distribution $\mu$ and events $\cE_1,\ldots,\cE_n$ with dependency graph $G$,
and any $\vec{p} \in \cS_G$ for which $\mu(\cE_i) \leq p_i$,
then $\mu\inp{\Intersect_{i=1}^n\overline{\cE_i}} > 0$,
and hence $\Intersect_{i=1}^n\overline{\cE_i} \neq \emptyset$.

In this section,
we give a new proof of Shearer's lemma in the so-called ``variable model'',
in which it is assumed that the events are determined by
underlying independent variables.
The variable model is assumed in most algorithmic formulations
of the Lov\'{a}sz Local Lemma~\cite{MoserTardos}.

For the purposes of this section, it will be slightly more convenient to use
the definition of the Shearer region from \Cref{def:ShearerRegion}:
\begin{equation}
\label{eq:OpenShearer}
\cS = \setst{ \vec{p} \in \bR^V }{ \qdown_S(\vec{p}) > 0 ~\:\forall S \subseteq V }.
\end{equation}

\paragraph{Preliminaries.}
We will use the variable model, as described in \Cref{app:VarMod}.
For simplicity we restrict to $\set{0,1}$-valued variables,
although similar arguments work for variables with arbitrary finite domains.
Given any vector $\vec{z} \in [0,1]^m$, let $\mu_{\vec{z}}$ now be the product
distribution on $\Omega = \set{0,1}^m$ with expectation $\vec{z}$.
We assume that event $\cE_i$ depends only on the coordinates $A_i \subseteq [m]$.
The dependency graph $G$ on vertex set $V=[n]$ has
an edge between $i$ and $j$ if $A_i \intersect A_j \neq \emptyset$.

\paragraph{Multilinearity.}
Let us now define $\vec{p} = \vec{p}(\vec{z}) \in [0,1]^n$ by $p_i = \mu_{\vec{z}}(\cE_i)$.
We first observe that each $p_i$ is a multilinear polynomial in $\vec{z}$:
\begin{equation}
\EquationName{PMultilinear}
p_i ~=~ p_i(\vec{z})
    ~=~ \sum_{S \in \Pi_{A_i}(\cE_i)} ~ \prod_{j \in S} z_j \prod_{j \in A_i \setminus S} (1-z_j),
\end{equation}
where $\Pi_{A_i}$ denotes projection to the coordinates in $A_i$.

The key observation is that
$\qdown_U( \vec{p}(\vec{z}) )$ is also a multilinear polynomial in $\vec{z}$,
for any $U \subseteq V$.
To see this, note that the events depending on variable $z_j$ form a 
clique in $G$, whereas each summand $\plusminus \prod_{i \in I} p_i$ in the definition of $\qdown_U$
involves an independent set $I$ in $G$.
Since a clique and an independent set intersect in at most one vertex,
each summand $\plusminus \prod_{i \in I} p_i$ involves at most one $p_i$ that depends on $z_j$.
So $\plusminus \prod_{i \in I} p_i$ is multilinear in $\vec{z}$, and the same is true of
$\qdown_U( \vec{p}(\vec{z}) )$.

\paragraph{Proof of Shearer's Lemma.}
In the variable model, the hypothesis of Shearer's lemma is that
there exists $\vec{z} \in [0,1]^m$ such that
under distribution $\mu_{\vec{z}}$ we have $\vec{p}(\vec{z}) \in \cS$.
The conclusion $\Intersect_{i=1}^n \overline{\cE_i} \neq \emptyset$
is equivalent to the existence of
an assignment $\vec{z} \in \set{0,1}^m$ to the variables with $\vec{p}(\vec{z})=\veczero$.

We prove that $\Intersect_{i=1}^n \overline{\cE_i} \neq \emptyset$ by the following argument.
Given an initial vector $\vec{z} \in [0,1]^m$ with $\vec{p}(\vec{z}) \in \cS$,
we first round $z_1$ to $0$ or $1$ while maintaining the property that $\vec{p}(\vec{z}) \in \cS$.
Then we repeat with $z_2,\ldots,z_m$, resulting in a final vector 
$\vec{z} \in \set{0,1}^m$ with $\vec{p} = \vec{p}(\vec{z}) \in \cS$.
As the distribution $\mu_{\vec{z}}$ is now deterministic, 
we must have $\vec{p} \in \set{0,1}^n$.
In fact $\vec{p}=\veczero$, for if $p_j = 1$ then $\qdown_{\set{j}} = 1-p_j = 0$,
contradicting that $\vec{p} \in \cS$.
Thus none of the events $\cE_1,\ldots,\cE_n$ occur under the assignment $\vec{z}$,
so $\vec{z} \in \Intersect_{i=1}^n \overline{\cE_i}$.

The crux is deciding how to round $z_i$.
We will increase $z_i$ to $1$ if
$\frac{\partial \qdown_V }{ \partial z_i }(\vec{p}(\vec{z})) \geq 0$,
and otherwise decrease $z_i$ to $0$.
This decision ensures that $\qdown_V$ does not decrease during this rounding step,
since $\qdown_V$ is multilinear in $\vec{z}$.
Note that the condition $\qdown_V(\vec{p})>0$ alone does not imply that $\vec{p} \in \cS$;
referring to \Cref{eq:OpenShearer}, we must also ensure that $\qdown_S(\vec{p})>0$
for all $S \subseteq V$.
Fortunately \Lemma{breve-prop} implies that
$\qdown_V(\vec{p}) = \min_{S \subseteq V} \qdown_S(\vec{p})$ whenever $\vec{p} \in \cS$.
So, thinking of \emph{continuously} modifying $z_i$, 
if any $\qdown_S(\vec{p})$ were to become non-positive then $\qdown_V(\vec{p})$
should be the first to do so.
Since the rounding is a continuous process ensuring that $\qdown_V(\vec{p})>0$,
this is actually sufficient to imply that $\vec{p} \in \cS$.
A formal version of this argument appears in \Cref{app:rigorousrounding}.

\paragraph{Remarks.}
Since this rounding argument for the LLL in Shearer's region is very simple,
one might be tempted to
try a similar rounding
in the region $\cL$ employed in the original form of the Lov\'asz Local Lemma
(see \Cref{app:LLL}).
It turns out that $\cL$ does not support such roundings: 
there is a probability space in the variable model with $\vec{p} \in \cL$,
such that rounding some variable to $0$ or $1$ will \emph{both} lead to $\vec{p} \not\in \cL$.
An example is shown in \Cref{app:Lrounding}.
So the fact that our rounding argument works is a special property of the Shearer region.

\medskip

This proof of Shearer's lemma directly suggests a potential algorithm:
compute the sign $\frac{\partial \qdown_V }{ \partial z_i }(\vec{p}(\vec{z}))$
in order to perform the rounding.
We design such an algorithm in the next section.

\subsection{An algorithmic LLL by polynomial evaluation}

We now explain how the algorithm from the previous subsection can be made
to run in subexponential time, using our FPTAS
\textsc{\CIP}, assuming that the initial
probabilities have slack.  We also assume that the probabilities
$p(\vec{z})$ of the events can be efficiently computed given the
probability distribution $\vec{z}$ on the underlying variables: this
is the case in standard applications of the variable model LLL such as
$k$-CNF-SAT.  The notation used in the theorem is as in the previous
subsection.

\begin{theorem}
Assume that the initial distribution $\vec{z}$ has slack $\alpha \in (0,1]$,
i.e., $(1+\alpha) \cdot \vec{p}(\vec{z}) \in \cS$, and that
$\vec{p}(\vec{z})$ can be computed from $\vec{z}$ in time polynomial
in $m$.
Then there is a deterministic algorithm with runtime
$(nm/\alpha)^{O(\log(d) \sqrt{m/\alpha})}$ that can construct a point
in $\Intersect_{i=1}^n \overline{\cE_i}$.  Here, $d$ is the degree of
the dependency graph.
\end{theorem}

\paragraph{Main ideas.}
Recall that the algorithm examines the sign of
$\frac{\partial \qdown_{[n]} }{ \partial z_i }(\vec{p}(\vec{z}))$
in order to decide whether to round $z_i$ up or down.
Since $\qdown_{[n]}$ is multilinear in $z_i$,
we can estimate $\frac{\partial \qdown_{[n]} }{ \partial z_i }(\vec{p}(\vec{z}))$
by using the FPTAS to compute $\qdown_{[n]}$ at two points nearby $\vec{p}(\vec{z})$.

Recall that the FPTAS is only efficient so long as there is
sufficiently large slack.
So the main challenge in the rounding is ensuring that the
points $\vec{p}(\vec{z})$ constructed during the algorithm
not only remain in $\cS$ but also have slack.
In each of the $m$ iterations, we might step a bit towards the Shearer
boundary, but we ensure that in one step, the slack cannot decrease by
more that $\frac{\alpha}{2m}$.  Since the initial slack is at least
$\alpha$, it can then be insured that all points constructed during
the $m$ iterations have slack at least $\frac{\alpha}{m}$.

\paragraph{Detailed discussion.}
The input to the algorithm is $\vec{z^0}$ satisfying $(1+\alpha) \cdot \vec{p}(\vec{z^0}) \in \cS$.
As a preprocessing step, we will first eliminate any coordinates of $\vec{z}$ that are
nearly integral:
if $z_i \leq \alpha/4$ we set $z_i$ to $0$,
and if $z_i \geq 1-\alpha/4$ we set $z_i$ to $1$.
In doing so, $p_j(\vec{z})$ can increase by at most a factor $(1-\frac{\alpha}{4})^{-1}$,
because $p_j$ is a probability and is multilinear in $z$.
So the resulting point $\vec{p}(\vec{z})$ still has slack at least $\frac{\alpha}{2}$.

As in the non-constructive version above, the algorithm has $m$
iterations, in which the $i^{\mathrm{th}}$ iteration rounds $z_i$ to
either $0$ or $1$.  Define $s_i = 1 + \frac{\alpha(m-i)}{2m}$. We maintain
the following invariant:
\ifbool{twocol}{
  \begin{equation}
\begin{gathered}
    \text{At the start of iteration $i$,}\\
    \text{the point $\vec{p}(\vec{z})$
  has slack at least $s_{i-1} - 1$.}
\end{gathered}\label{eq:13}\tag{$\dagger$}
\end{equation}

}{
  \begin{equation}
\text{At the start of iteration $i$,  the point $\vec{p}(\vec{z})$
  has slack at least $s_{i-1} - 1$.} \label{eq:13}\tag{$\dagger$}
\end{equation}
}
We then proceed to estimate $\frac{\partial \qdown_{[n]} }{ \partial z_i }$
at the point $s_i \cdot \vec{p}(\vec{z})$ which, due to the definition of $s_i$,
still has slack $\frac{\alpha}{4m}$.
Note that
\begin{equation}
\label{eq:qderiv}
\frac{\partial \qdown_{[n]} }{ \partial z_i }
\big( s_i \cdot \vec{p}(\vec{z}) \big)
~=~ \frac{1}{\delta} \Big(
      \qdown_{[n]}\big( s_i \cdot \vec{p}(\vec{z}) \big)
    - \qdown_{[n]}\big( s_i \cdot \vec{p}(\vec{z}-\delta \vec{e}_i) \big)
    \Big),
\end{equation}
by multilinearity.
Our algorithm will choose an appropriate $\delta$, then use the FPTAS to
estimate the two terms on the right-hand side with sufficiently small error.

First we must check that $s_i \cdot \vec{p}(\vec{z}-\delta \vec{e}_i)$
is still in the Shearer region, and with sufficient slack.  Set
$\delta=\frac{\alpha^2}{36m}$.  Since $z_i \geq \alpha/4$, we
certainly have $\vec{p}(\vec{z}-\delta \vec{e}_i) \geq 0$.  Next we
will prove that
\begin{equation}
\label{eq:coordinatewise}
\textstyle
s_i \cdot \vec{p}(\vec{z}-\delta \vec{e}_i) ~\leq~ (1 +
\frac{\alpha}{9m}) s_i \cdot \vec{p}(\vec{z})
\end{equation}
coordinate-wise.
Since $s_i \cdot \vec{p}(\vec{z})$ has slack $\frac{\alpha}{4m}$,
the inequality \eqref{eq:coordinatewise} will then imply that
$s_i \cdot \vec{p}(\vec{z}-\delta \vec{e}_i)$ has slack $\frac{\alpha}{8m}$.

Let us consider the $j$th coordinate in \eqref{eq:coordinatewise}.
Fixing all coordinates of $\vec{z}$ other than $z_i$,
we may write $p_j$ as the linear function $\mu z_i + \nu$,
where $\mu+\nu \geq 0$ since $p_j$ is a probability.
We may assume that $\mu \leq 0$, otherwise the inequality \eqref{eq:coordinatewise} is trivial.
Then, using $\mu \leq 0$ and recalling that $z_i \leq 1 - \frac{\alpha}{4}$, we have
\ifbool{twocol}{
  \begin{multline*}
    \frac{\mu(z_i-\delta)+\nu}{\mu z_i+\nu} ~=~ 1 - \frac{\mu
      \delta}{\mu z_i + \nu} ~\leq~ 1 - \frac{\mu \delta}{\mu
      (1-\frac{\alpha}{4}) + \nu} \\
    ~\leq~ 1 + \frac{\delta}{\alpha/4}
    ~=~ 1 + \frac{\alpha}{9m}.
  \end{multline*}
}{
  $$
\frac{\mu(z_i-\delta)+\nu}{\mu z_i+\nu}
~=~    1 - \frac{\mu \delta}{\mu z_i + \nu}
~\leq~ 1 - \frac{\mu \delta}{\mu (1-\frac{\alpha}{4}) + \nu}
~\leq~ 1 + \frac{\delta}{\alpha/4}
~=~    1 + \frac{\alpha}{9m}.
$$
}
This proves \eqref{eq:coordinatewise}.

Now let us return to our estimation of \eqref{eq:qderiv}.
For simplicity let us rewrite \eqref{eq:qderiv} with the following shorthand notation
\ifbool{twocol}{
  \begin{multline*}
    \ifbool{twocol}{\!\!\!\!\!\!\!\!\!}{}
    \underbrace{\frac{\partial \qdown_{[n]} }{ \partial z_i } \big(
      s_i \cdot \vec{p}(\vec{z}) \big)}_{a}
      \ifbool{twocol}{=}{~=~}
      \frac{1}{\delta} \Big(
    \underbrace{\qdown_{[n]}\big( s_i \cdot \vec{p}(\vec{z})
      \big)}_{q_0} - \underbrace{\qdown_{[n]}\big( s_i \cdot
      \vec{p}(\vec{z}-\delta \vec{e}_i) \big)}_{q_\delta} \Big)
    \\~~ \iff ~~ a ~=~ \frac{q_0-q_\delta}{\delta}.
  \end{multline*}
}{\[
\underbrace{\frac{\partial \qdown_{[n]} }{ \partial z_i }
\big( s_i \cdot \vec{p}(\vec{z}) \big)}_{a}
~=~ \frac{1}{\delta} \Big(
      \underbrace{\qdown_{[n]}\big( s_i \cdot \vec{p}(\vec{z}) \big)}_{q_0}
    - \underbrace{\qdown_{[n]}\big( s_i \cdot \vec{p}(\vec{z}-\delta \vec{e}_i) \big)}_{q_\delta}
    \Big)
\quad~~
\iff
\quad~~
a ~=~ \frac{q_0-q_\delta}{\delta}.
\]}
We set $\epsilon = \frac{\delta}{4}$, then use the FPTAS to compute quantities
$Q_0 \in [1-\epsilon,1] \cdot q_0$ and
$Q_\delta \in [1-\epsilon,1] \cdot q_\delta$,
then compute
$$
A ~=~ \frac{Q_0 - Q_\delta}{\delta}.
$$
It then follows that
\begin{equation}
\label{eq:Aadiff}
-\frac{\epsilon}{\delta} q_\delta ~\leq~ a-A ~\leq~ \frac{\epsilon}{\delta} q_0.
\end{equation}

The algorithm's rounding proceeds as follows. Let
$\vec{z'} = \vec{z}$.  If $A \geq 0$ then round $z'_i$ to $1$,
otherwise round $z'_i$ to $0$.  We claim that $\vec{z'}$ satisfies
$s_i \cdot \vec{p}(\vec{z'}) \in \cS$. This implies that $\vec{z'}$
has slack at least $s_i - 1$, so that the invariant \eqref{eq:13} is
satisfied at the start of the next (i.e., the $(i+1)$th) iteration.

To prove this claim, we begin by observing that \Lemma{continuous} implies that it
is sufficient to prove that while rounding the $i$th co-ordinate,
$\qdown_{[n]}( s_i \cdot \vec{p}(\vec{z''}) )$ is strictly positive on all
points $\vec{z''}$ lying on the straight line segment along which
$z_i$ is rounded.  Let $\vec{z''}$ be a point on this line segment.
Consider first the case $A < 0$.  Using \cref{eq:Aadiff}, we have
\ifbool{twocol}{
  \begin{align*}
    \textstyle \qdown_{[n]}( s_i \cdot \vec{p}(\vec{z''}) ) ~=~ q_0 -
    (z_i- z_i'') a
    &~\geq~ q_0 (1 - z_i \frac{\epsilon}{\delta} )\\
    &~\geq~ q_0 / 2 ~>~ 0.
  \end{align*}
}{\[
\textstyle
\qdown_{[n]}( s_i \cdot \vec{p}(\vec{z''}) )
~=~ q_0 - (z_i- z_i'') a
~\geq~ q_0 (1 - z_i \frac{\epsilon}{\delta} )
~\geq~ q_0 / 2 ~>~ 0.
\]
}
Next consider the case $A \geq 0$.
Then $0 \leq Q_0 - Q_\delta \leq q_0 - (1-\epsilon) q_\delta$,
and hence $q_\delta \leq (1+2\epsilon) q_0 \leq 2 q_0$.
Then, by eq.~\eqref{eq:Aadiff}, we have
$ a \geq - \frac{\epsilon}{\delta} q_\delta
  \geq - \frac{2 \epsilon}{\delta} q_0 $.
It follows that
\ifbool{twocol}{
  \begin{align*}
    \textstyle \qdown_{[n]}( s_i \cdot \vec{p}(\vec{z''}) ) ~=~ q_0 +
    (z_i''-z_i) a &~\geq~ q_0 \big( 1 - (1-z_i) \frac{2 \epsilon }{
      \delta } \big)\\
    &~\geq~ q_0 / 2 ~>~ 0.
  \end{align*}
}{
  $$
  \textstyle
\qdown_{[n]}( s_i \cdot \vec{p}(\vec{z''}) )
 ~=~ q_0 + (z_i''-z_i) a
~\geq~ q_0 \big( 1 - (1-z_i) \frac{2 \epsilon }{ \delta } \big)
~\geq~ q_0 / 2 ~>~ 0.
$$
}
This completes the claim that $s_i \cdot \vec{p}(\vec{z'}) \in \cS$.

\paragraph{Runtime.}
The FPTAS is invoked $O(m)$ times, each time with
$\epsilon = \frac{\delta}{4} = \frac{\alpha^2}{144m}$
and with slack at least $\alpha/8m$.
The runtime is therefore
\ifbool{twocol}{
  \begin{multline*}
    O(m) \cdot \Big(\frac{n}{\epsilon (\alpha/8m)}
    \Big)^{O\big(\log(d)/\sqrt{\alpha/8m}\big)}
    \\~=~
    \Big(\frac{nm}{\alpha} \Big)^{O\big(\log(d)\sqrt{m/\alpha}\big)}.
  \end{multline*}
}{
  $$
O(m) \cdot \Big(\frac{n}{\epsilon (\alpha/8m)} \Big)^{O\big(\log(d)/\sqrt{\alpha/8m}\big)}
~=~ \Big(\frac{nm}{\alpha} \Big)^{O\big(\log(d)\sqrt{m/\alpha}\big)}.
$$
}

%

%

%

%

%

%

%

%

%

%

%


\section{Testing membership in Shearer's region}
\label{sec:membership}

In this section we consider the following question:
\begin{question*}
  Given a graph $G$ and activities $p_v$ for all $v \in V$, is $\vec{p}$
  in the Shearer region of $G$?
\end{question*}
We recall that $\vec{p} \in {\cal S}_G$ if and only if $\breve{q}_S(\vec{p}) > 0$ for all $S \subseteq [n]$,
or equivalently if $\breve{q}_V(\vec{p}) > 0$ everywhere on the line segment connecting $\vec{0}$ and $\vec{p}$.
As we prove, it is \#P-hard to answer this question exactly (see~\Cref{app:hardness}).
On the other hand, in running time $2^{O(n)}$, we can trivially compute all Shearer's polynomials $\breve{q}_S(\vec{p})$ and answer this question. Here we show that we can test membership approximately in subexponential time.

\subsection{An algorithm to test membership in Shearer's region}

\begin{theorem}
  \TheoremName{membership}
  Given a graph $G$, $p_v \in (0,1)$ for $v \in V$, and
  $0<\alpha<1$, there exists a deterministic algorithm which, in running time
  $(n/\alpha)^{O(\sqrt{n / \alpha} \log d)}$, decides whether
  $\vec{p} \in {\cal S}_G$ or $(1+\alpha) \vec{p} \notin {\cal S}_G$.
\end{theorem}

\medskip

\emph{Notation.} It will be convenient to express some of the
arguments in this section in terms of Shearer polynomials $q$ defined
as follows (see, e.g.~\cite{HV15}).  For any subset $S$ of vertices in
a graph $G = (V, E)$,
\begin{equation}
  \label{eq:15}
  q_S(\vec{p}) \defeq \sum_{\substack{I \in \Ind(G)\\I \supseteq
      S}}(-1)^{|I\setminus S|}\prod_{i \in S}p_i.
\end{equation}
Note that these polynomials are related to the $\breve{q}$ as follows:
\begin{displaymath}
  q_S(\vec{p}) =
  \begin{cases}
    \inp{\prod_{i \in S}p_i} \breve{q}_{V \setminus \Gamma^+(S)} & \text{ if
      $S \in \Ind(G)$,}\\
    0 &\text{ otherwise }.
  \end{cases}
\end{displaymath}

We want to harness our evaluation algorithm for $q_\emptyset(\vec{p}) = \breve{q}_V(\vec{p})$.
However, we need to proceed carefully, since evaluating $q_\emptyset({\vec p})$ without knowing a lower bound on the slack might give the wrong answer. Hence we start from a point which is guaranteed to be in the Shearer region, and we maintain a lower bound on the slack at each point. We will use the following lemma.

\begin{lemma}
\LemmaName{slack-test}
For a point $\vec{p} \in {\cal S}_G$, let $\gamma(\vec{p}) = 1 / \sum_{i=1}^{n} \frac{q_{\{i\}}(\vec{p})}{q_\emptyset(\vec{p})}.$ Then the slack of $\vec{p}$ is between $\gamma(\vec{p})$ and $n \gamma(\vec{p})$, or in other words: $(1+\gamma(\vec{p})-\epsilon) \vec{p} \in {\cal S}_G$ for every $\epsilon>0$, and $(1+n\gamma(\vec{p})) \vec{p} \notin {\cal S}_G$.
\end{lemma}

\begin{proof}
We know from \cite{HV15} that $q_\emptyset((1+t) \vec{p})$ is a convex function of $t$, and
\ifbool{twocol}{
  \begin{multline*}
    \frac{d}{dt} q_\emptyset((1+t) \vec{p}) \Big|_{t=0} =
    \sum_{i=1}^{n} p_i \frac{\partial q_\emptyset}{\partial
      p_i}(\vec{p}) = -\sum_{i=1}^{n} p_i \breve{q}_{V \setminus
      \Gamma^+(i)}(\vec{p})\\
    = -\sum_{i=1}^{n} q_{\{i\}}(\vec{p}).
  \end{multline*}
}{
  $$ \frac{d}{dt} q_\emptyset((1+t) \vec{p}) \Big|_{t=0} = \sum_{i=1}^{n} p_i \frac{\partial q_\emptyset}{\partial p_i}(\vec{p})
= -\sum_{i=1}^{n} p_i \breve{q}_{V \setminus \Gamma^+(i)}(\vec{p}) =
-\sum_{i=1}^{n} q_{\{i\}}(\vec{p}).$$
}
Therefore, for $\lambda < \gamma(\vec{p}) = 1 / \sum_{i=1}^{n} \frac{q_{\{i\}}(\vec{p})}{q_\emptyset(\vec{p})}$,
we have
\ifbool{twocol}{
  \begin{align*}
    \breve{q}_V((1+\lambda)\vec{p})
    &= q_\emptyset((1+\lambda)\vec{p})
    \\
    &\geq q_\emptyset(\vec{p}) + \lambda \frac{d}{dt} q_\emptyset((1+t)
    \vec{p}) \Big|_{t=0}
    \\
    &= q_\emptyset(\vec{p}) - \lambda \sum_{i=1}^{n} q_{\{i\}}(\vec{p}) > 0.
  \end{align*}
}{
  $$\breve{q}_V((1+\lambda)\vec{p}) = q_\emptyset((1+\lambda)\vec{p})
\geq q_\emptyset(\vec{p}) + \lambda \frac{d}{dt} q_\emptyset((1+t)
\vec{p}) \Big|_{t=0} = q_\emptyset(\vec{p}) - \lambda \sum_{i=1}^{n}
q_{\{i\}}(\vec{p}) > 0.$$
}
Along with \Cref{lem:continuous}, this implies that
$(1+\gamma(\vec{p})-\epsilon) \vec{p} \in {\cal S}_G$ for every
$\epsilon>0$, since $\vec{p} \in {\cal S}_G$ and
$q_\emptyset((1+\lambda) \vec{p}) > 0$ for every
$\lambda \in [0,\gamma(\vec{p}))$.

On the other hand, we have $n \gamma(\vec{p}) \geq \frac{q_\emptyset(\vec{p})}{q_{\{i^*\}}(\vec{p})}$ where
$q_{\{i^*\}}(\vec{p}) = \max_i q_{\{i\}}(\vec{p})$. Since $q_\emptyset(\vec{p})$ is linear in $p_{i^*}$,
we obtain
$$ q_\emptyset\left(\vec{p} + \frac{q_\emptyset(\vec{p})}{q_{\{i^*\}}(\vec{p})} p_{i^*} {\mathbf e}_{i^*} \right)
= q_\emptyset(\vec{p}) + p_{i^*}
\frac{q_\emptyset(\vec{p})}{q_{\{i^*\}}(\vec{p})}
\frac{\partial{q_\emptyset}}{\partial p_{i^*}} = 0,$$ since
$p_i\frac{\partial{q_\emptyset}}{\partial p_{i^*}} = -
q_{\inb{i^*}}(\vec{p})$.  Thus, by monotonicity of the Shearer region,
$(1+n\gamma(\vec{p})) \vec{p} \notin {\cal S}_G$.
\end{proof}

Now we can prove \Theorem{membership}.

\begin{proof}
  We maintain a point $\vec{\tilde{p}}$ such that
  $(1+\frac{\alpha}{8n}) \vec{\tilde{p}} \in {\cal S}_G$.  We start
  with $\vec{\tilde{p}} = \frac{1}{2n}\cdot \vec{p}$ (where $\vec{p}$
  is the input vector) for which the statement is certainly true
  (since $(\frac{1}{n}, \ldots, \frac{1}{n}) \in {\cal S}_G$ for any
  graph $G$).  Note further that this initial point has slack $O(n)$,
  since if each $p_v$ was less than $1/n$, $\vec{p} \in \cS$ is
  trivially true.

Given $\vec{\tilde{p}}$, we compute $q_\emptyset(\vec{\tilde{p}})$ and $q_{\{i\}}(\vec{\tilde{p}})$ for all $i$ within a relative error of $\epsilon = 1/2$. More precisely, we obtain estimates within $[1, 3/2]$ times the correct value.
Thus we can also compute $\gamma(\vec{\tilde{p}}) = 1 / \sum_{i=1}^{n} \frac{q_{\{i\}}(\vec{\tilde{p}})}{q_\emptyset(\vec{\tilde{p}})}$
within $[2/3,3/2]$ times the correct value. Let us call this estimate $\tilde{\gamma}$. If $\tilde{\gamma} \leq \frac{\alpha}{2n}$ then we can conclude by \Lemma{slack-test} that $(1+\alpha) \vec{\tilde{p}} \notin {\cal S}_G$.

If on the other hand $\tilde{\gamma} > \frac{\alpha}{2n}$, we know by \Lemma{slack-test} that $(1 + \frac{\alpha}{3n}) \vec{\tilde{p}} \in {\cal S}_G$, and hence $(1 + \frac{\alpha}{6n}) \vec{\tilde{p}}$ still has a slack of say $\frac{\alpha}{8n}$.
In this case we replace $\vec{\tilde{p}}$ by $(1 + \frac{\alpha}{6n}) \vec{\tilde{p}}$ and continue.

Let us analyze the running time. Whenever we evaluate $q_\emptyset(\vec{\tilde{p}}) = \breve{q}_V(\vec{\tilde{p}})$ and $q_{\{i\}}(\vec{\tilde{p}}) = p_i \breve{q}_{V \setminus \Gamma^+(i)}(\vec{\tilde{p}})$, we have a guaranteed slack of $\Omega(\alpha/n)$. By \Theorem{mainShearer}, we can evaluate these polynomials within a relative error of $\epsilon = 1/2$ in running time $(\frac{2n}{\alpha})^{4 \sqrt{n/\alpha} \log d} = (n/\alpha)^{O(\sqrt{n / \alpha} \log d)}$. Every time we iterate, the actual slack goes down by a factor of $1 - \Omega(1/n)$ (from \Lemma{slack-test}). Since the multiplicative slack at the beginning is $O(n)$, and we stop when the slack becomes $O(\alpha/n)$, the number of iterations is $O(n \log (n/\alpha))$. Thus the total running time is $(n/\alpha)^{O(\sqrt{n / \alpha} \log d)}$.
\end{proof}


%
%
%

\section{Conclusions and open questions}
\label{sec:concl}
The main open question left open by our work is whether the dependence
on the slack in \Theorem{mainShearerIntro} can be improved from
sub-exponential to polynomial.  In all our applications of
\Theorem{mainShearerIntro}, we have to work with sub-constant values
of the slack $\alpha$, and it is the sub-exponential dependence on
$1/\alpha$ that prevents us from giving polynomial time algorithms for
these applications.

\begin{question}
    Is there an algorithm to estimate $\qdown_V(\vec{p})$ up to a $(1+\epsilon)$-multiplicative
    factor in $n$-vertex graphs of maximum degree $d$,
    assuming that $(1+\alpha) \vec{p} \in \mathcal{S}$,
    in running time $(\frac{n}{\alpha \epsilon})^{O(\log d)}$?
\end{question}
In \Cref{sec:optim-decay-rate}, we provide some evidence to show
that the correlation decay technique may not be capable of completely
removing the sub-exponential dependence on $1/\alpha$.  On the other
hand, the result of Patel and Regts~\cite{PatelRegts}, to the best of
our knowledge, has an even worse exponential dependence on $1/\alpha$,
which, as discussed in \Cref{sec:relat-work-techn}, appear to be
intrinsic to the techniques used there. %
A ray of hope is however offered by some recent progress in the
positive activity setting, where Efthymiou, Hayes,
\v{S}tefankovi\v{c}, Vigoda and Yin~\cite{efthymiou2016convergence}
were able to obtain an FPRAS for the independence polynomial on graphs
of large enough bounded degree and large enough girth, the exponent of
whose running time has no dependence on the analog of the slack in the
positive activity setting (the slack only appears in the time
complexity of their algorithm as a multiplicative factor).  The
starting point of their result is the tight connection between
approximation of the independence polynomial at positive activities
and sampling from the associated \emph{Gibbs distribution}; they then
exploit this connection by showing that a natural Markov chain can
sample efficiently from the Gibbs distribution.  Their proof uses the
connection between correlation decay and the mixing properties of
Markov chains for the Gibbs distribution in an novel interesting
fashion.

Unfortunately, to the best of our knowledge, no natural analogue of
the Gibbs distribution is known for the negative activities setting.
Thus, it remains an open problem to find if, and how, sampling
techniques such as Markov chain Monte Carlo can be brought to bear
upon the above question.  Moreover, for applications to the LLL in the
variable model, one would also need to remove the large girth
assumption that appears to be crucial to the result of Efthymiou et
al.~\cite{efthymiou2016convergence}.

%

%

%


%
\clearpage
\appendix
\section{The Lov\'{a}sz Local Lemma and Shearer's Lemma}
\label{app:LLL}

The Lov\'asz Local Lemma (LLL) is a fundamental tool used in combinatorics
to argue that the probability that none of a set of suitably
constrained bad events occurs is positive.  In abstract terms, the
lemma is formulated in terms of $n$ events $\cE_1, \cE_2, \dots, \cE_n$
and a probability distribution $\mu$ on the events.  However, only two
pieces of data about the distribution $\mu$ are used in the
formulation of lemma:
\begin{itemize}
\item The marginal probabilities $p_i \defeq \mu(\cE_i)$ of the
  events, and
\item A \emph{dependency graph} $G = (V, E)$ associated with $\mu$.
  The vertices $V$ are identified with the
  events $\cE_1, \cE_2, \dots, \cE_n$, and the graph is interpreted as
  stipulating that under the distribution $\mu$, the event $\cE_i$ is
  independent of its non-neighbors in the graph $G$.
\end{itemize}
The original LLL \cite{ErdosLovasz,Spencer77} provides \emph{sufficient} conditions on the $p_i$
and the dependency graph $G$ that ensure that $\mu\inp{\Intersect_{i=1}^n\overline{\cE_i}}>0$,
and hence $\Intersect_{i=1}^n\overline{\cE_i} \neq \emptyset$.

Define the set
\ifbool{twocol}{
\begin{equation*}
\cL
 ~=~ \{ \vec{p} \in [0,1]^n : \exists \vec{x} \in (0,1)^n ~\text{s.t.}~ 
             p_i \leq x_i \prod_{j \in \Gamma(i)} (1-x_j) \}.
\end{equation*}
}{
\begin{equation*}
\cL
 ~=~ \setst{ \vec{p} \in [0,1]^n }{ \exists \vec{x} \in (0,1)^n ~\text{s.t.}~ 
             p_i \leq x_i \prod_{j \in \Gamma(i)} (1-x_j) }.
\end{equation*}
}
\begin{theorem}[The Lov{\'{a}}sz Local Lemma \cite{ErdosLovasz,Spencer77}]
\TheoremName{LLL}
If $\vec{p} \in \cL$ then $\mu\inp{\Intersect_{i=1}^n\overline{\cE_i}} > 0$.
\end{theorem}

Shearer's remarkable lemma~\cite{Shearer} provides \emph{necessary and sufficient}
conditions for $\mu(\Intersect_{i=1}^n \overline{\cE_i})>0$
to hold for a given dependency graph $G$ and vector of probabilities $\vec{p} \in [0,1]^V$.
Scott and Sokal~\cite{ScottSokal} showed that Shearer's conditions can be
expressed very succinctly in the language of partition functions.
Recall the definition of $\cS$ from Section~\ref{sec:our-results}. 

\begin{theorem}[Shearer's Lemma \cite{Shearer,ScottSokal}]
\TheoremName{Shearer}
If $\vec{p} \in \cS$ then $\mu\inp{\Intersect_{i=1}^n\overline{\cE_i}} > 0$.

Conversely, if $\vec{p} \not\in \cS$ then there exists a distribution $\mu$
with dependency graph $G$ satisfying $\mu(\cE_i) = p_i$
such that $\mu\inp{\Intersect_{i=1}^n\overline{\cE_i}} = 0$.
\end{theorem}

The fact that $\cL \subseteq \cS$ follows indirectly from the statements of
Theorems~\ref{thm:LLL} and \ref{thm:Shearer}, 
but a direct argument is also known \cite[Corollary 5.37]{HV15}.

\subsection{The Variable Model}
\label{app:VarMod}

In order to design algorithmic forms of the LLL, some assumption must be made
on the probability space. 
The most natural assumption is the ``variable model'',
in which the probability space consists of $m$ independent variables with finite domain,
each event depends on some subset of the variables,
and that two events are adjacent in $G$
if there is a common variable on which both depend.

This paper will focus on the specific scenario where the variables take values in $\set{0,1}$.
Concretely, the probability space is supported on $\Omega = \set{0,1}^m$.
Its distribution is $\mu_{\vec{z}}$, the product distribution on $\Omega$
with marginal vector $\vec{z}$.
That is, if $\vec{\omega}$ has distribution $\mu_{\vec{z}}$ then $\E[]{\vec{\omega}}=\vec{z}$.
Each event is a Boolean function of $\vec{\omega}$ that only depends on
certain coordinates $A_i \subseteq [m]$.
Let $G$ be the dependency graph on $V=[n]$
where $i$ and $j$ are adjacent if $A_i \intersect A_j \neq \emptyset$.
Since $\mu_{\vec{z}}$ is a product distribution,
each event $\cE_i$ is clearly independent from its non-neighbors in $G$.

%

\section{Proofs from \Section{preliminaries}}

\begin{proof}[Proof of \Cref{clm:selfreduc}]
  From \cref{eq:11}, we have, for every $1 \leq i \leq n$,
  \ifbool{twocol}{$  1 - r_{S_i, v_i} =
    \frac{\qdown_{S_i}}{\qdown_{S_{i+1}}}.$ }{
  \begin{displaymath}
    1 - r_{S_i, v_i} = \frac{\qdown_{S_i}}{\qdown_{S_{i+1}}}.
  \end{displaymath}
}
Multiplying these equations, we get
  \begin{displaymath}
    \prod_{i=1}^n(1 - r_{S_i, v_i}) =
    \frac{\qdown_{S_1}}{\qdown_{S_{n+1}}} = \frac{\qdown_V}{\qdown_\emptyset}.
  \end{displaymath}
  Since $\qdown_\emptyset = 1$, this yields the claim.
\end{proof}

\begin{proof}[Proof of \Cref{lem:recurrence}]
  Define $S_i = S\setminus\inb{u, v_1, v_2, \dots, v_{i-1}}$ for
  $1 \leq i \leq k$.  From \cref{eq:11}, we then have, for
  $1 \leq i \leq k$,
  \ifbool{twocol}{$\frac{1}{1 - r_{S_i, v_i}} =
    \frac{\qdown_{S_{i+1}}}{\qdown_{S_i}}.$ }{
  \begin{displaymath}
    \frac{1}{1 - r_{S_i, v_i}} = \frac{\qdown_{S_{i+1}}}{\qdown_{S_i}}.
  \end{displaymath}
  } Multiplying these equations, we get
  \begin{displaymath}
    \prod_{i=1}^k\frac{1}{1 - r_{S_i, v_i}} =
    \frac{\qdown_{S_{k+1}}}{\qdown_{S_1}} = \frac{\qdown_{S \setminus
        \Gamma^+(u)}}{\qdown_{S\setminus\inb{u}}},
  \end{displaymath}
  where in the last equation we use $S_{k+1} = {S \setminus \Gamma^+(u)}$
  and $S_1 = S \setminus \inb{u}$. (Recall that $\Gamma^+(u)$ is the
  set containing $u$ and all its neighbors in $G$).  The claim of the
  lemma now follows since $r_{S,u} = p_u\cdot \frac{\qdown_{S \setminus
      \Gamma^+(u)}}{\qdown_{S\setminus\inb{u}}}$ (see \cref{eq:11}).
\end{proof}

\begin{proof}[Proof of \Cref{lem:breve-prop}]
  This is a special case of Corollary 2.27~(b) of Scott and Sokal~\cite{ScottSokal}.
  See also \cite[Section 5.3]{HV15}.
\end{proof}

\begin{proof}[Proof of \Cref{lem:ratios-bounded}]
  From \cref{eq:11} we have
  $r_{S, u} =
  \frac{p_u\qdown_{S\setminus\Gamma^+(u)}}{\qdown_{S\setminus\inb{u}}}
  = 1 - \frac{\qdown_S}{\qdown_{S\setminus \inb{u}}}$.  From
  \Cref{lem:breve-prop}, we have
  $0 < \qdown_V \leq \qdown_S \leq \qdown_{S \setminus \inb{u}} \leq
  \qdown_{S\setminus\Gamma^+(u)}$, which yields the claim.
\end{proof}

\section{Inequalities in the complex plane}
\label{sec:inequalities}

Finally, we enumerate some simple inequalities that are used in our
proofs.
\begin{fact}\label{fct:simple}
  Let $z$ be a complex number such that $\abs{z} \leq \tau < 1$.  Then
  $\frac{1}{\abs{1 - z}} \leq \frac{1}{1 - \tau}$.
\end{fact}

\begin{proof}
  $\abs{1 -z} \geq 1 - \abs{z} \geq 1 - \tau$, which implies the
  claim since $\tau < 1$.
\end{proof}

\begin{fact}\label{fct:multiply}
  Let $\inp{x_i}_{i=1}^n$ and $\inp{y_i}_{i=1}^n$ be two sequences of
  complex numbers with the $y_i$ non-zero such that
  \begin{equation}
    \label{eq:1} \abs{\frac{x_i}{y_i} - 1} \leq \epsilon,
  \end{equation}
  where $\epsilon \leq 1/n$.  Then, we have
  \begin{displaymath}
    \abs{\prod_{i=1}^n x_i - \prod_{i=1}^n y_i} \leq
    2 n \epsilon \cdot \prod_{i=1}^n\abs{y_i}.
  \end{displaymath}
\end{fact}

\begin{proof}
  For each $i$, define $z_i$ so that $x_i = y_i(1 + z_i)$. Note that
  $\abs{z_i} \leq \epsilon$ for each $i$.  We therefore have
  \ifbool{twocol}{
    \begin{align*}
    \abs{\prod_{i=1}^n x_i - \prod_{i=1}^n y_i}
    &= \prod_{i=1}^n\abs{y_i}\abs{\prod_{i=1}^n (1 + z_i) - 1}\\
    &\leq \sum_{i=1}^n\binom{n}{i}\epsilon^i
      \cdot \prod_{i=1}^n\abs{y_i}\\
    &\leq n\epsilon \sum_{i=1}^n\frac{(n \epsilon)^{i-1}}{i!}
      \cdot \prod_{i=1}^n\abs{y_i}
      \text{, using $\binom{n}{i} \leq \frac{n^i}{i!}$},\\
    &\leq \prod_{i=1}^n\abs{y_i} \cdot n\epsilon
      \sum_{i=1}^n\frac{1}{i!}\\
    &\leq n\epsilon\cdot (e -
      1) \cdot \prod_{i=1}^n\abs{y_i}
      \leq 2 n\epsilon \cdot \prod_{i=1}^n\abs{y_i},
  \end{align*}
  }{
  \begin{align*}
    \abs{\prod_{i=1}^n x_i - \prod_{i=1}^n y_i}
    &= \prod_{i=1}^n\abs{y_i}\abs{\prod_{i=1}^n (1 + z_i) - 1}
      \leq \sum_{i=1}^n\binom{n}{i}\epsilon^i
      \cdot \prod_{i=1}^n\abs{y_i}\\
    &\leq n\epsilon \sum_{i=1}^n\frac{(n \epsilon)^{i-1}}{i!}
      \cdot \prod_{i=1}^n\abs{y_i}
      \text{, using $\binom{n}{i} \leq \frac{n^i}{i!}$},\\
    &\leq \prod_{i=1}^n\abs{y_i} \cdot n\epsilon
      \sum_{i=1}^n\frac{1}{i!}
      \leq n\epsilon\cdot (e -
      1) \cdot \prod_{i=1}^n\abs{y_i}
      \leq 2 n\epsilon \cdot \prod_{i=1}^n\abs{y_i},
  \end{align*}
  }
  where the last line uses the fact that $n \epsilon < 1$.
\end{proof}

We will also need the following consequence of the mean value theorem.
Fix a complex number $\lambda$ and a positive integer $d > 0$ and let
$f(\vec{x}) = f(x_1, x_2, \dots, x_d)$ be defined as
\begin{displaymath}
  f(x_1, x_2, \dots, x_d) \defeq \lambda \prod_{i=1}^d\frac{1}{1 - x_i}
\end{displaymath}
defined when $\abs{x_i} < 1$ for all $i$.

\begin{theorem}[\textbf{Mean value theorem}]
  \label{thm:mean-value}
  Let $\vec{x} = (x_1, x_2, \dots, x_d)$ and $\vec{y} = (y_1, y_2, \dots, y_d)$ be two
  sequences of complex numbers and let
  $\vec{\gamma} = (\gamma_1, \gamma_2, \dots, \gamma_d)$ be such that
  $\abs{x_i}, \abs{y_i} \leq \gamma_i < 1$ for $1 \leq i \leq d$.
  Then
  \begin{displaymath}
    \abs{f(\vec{x}) - f(\vec{y})} \leq \abs{f(\vec{\gamma})}
    \sum_{i = 1}^d\frac{\abs{x_i - y_i}}{1 - \gamma_i}.
  \end{displaymath}
\end{theorem}

\begin{proof}
  Let $g(t) \defeq f(t\vec{x} + (1-t)\vec{y})$ for $t \in [0,1]$.  Note that $g$
  is continuously differentiable on its domain (since
  $\abs{x_i}, \abs{y_i} < 1$). Hence $\abs{g'(t)}$ attains its maximum
  at some point $t_0 \in [0,1]$. Let $\vec{z} = t_o \vec{x} + (1-t_0)\vec{y}$. Note that
  $\abs{z_i} \leq \gamma_i$ for all $i$. We now have
  \begin{align*}
    \abs{f(\vec{x}) - f(\vec{y})}
    &= \abs{g(1) - g(0)} \leq \int_0^1\abs{g'(t)} dt \leq
      \abs{g'(t_0)}\\
    &= \abs{f(\vec{z})}\abs{\sum_{i=1}^d\frac{y_i - x_i}{1 - z_i}}\\
    & \leq \abs{f(\vec{\gamma})}\sum_{i=1}^d\frac{\abs{x_i - y_i}}{1 -
      \gamma_i},
  \end{align*}
  where the last line uses the form of $f$ and \Cref{fct:simple}.
\end{proof}

\section{Proofs from \Cref{sec:ShearerVar}}

\subsection{Correctness of the rounding procedure}
\label{app:rigorousrounding}

Define the following set
$$
\cR ~=~ \setst{ \vec{p} \in [0,1]^n }{ \qdown_{V}(\vec{p})>0 }.
$$

\begin{lemma}
\LemmaName{continuous}
Let $f : [0,1] \rightarrow \cR$ be a continuous function with $f(0) \in \cS$.
Then $f(t) \in \cS$ for all $t \in [0,1]$.
\end{lemma}

\begin{proof}
Suppose that $\setst{ t \in [0,1] }{ f(t) \not\in \cS }$ is non-empty,
and let $\tau$ be the infimum of that set.
By continuity of $f$ and the fact that $\cS$ is open,
we must have $f(\tau) \not\in \cS$.
That is, there exists $S \subseteq V$ with $0 \geq \qdown_S(f(\tau))$.
As $f(0) \in \cS$, we have $\tau > 0$.
The range of $f$ is $\cR$, so $\qdown_{V}(f(\tau)) > 0 \geq \qdown_S(f(\tau))$.
For sufficiently small $\epsilon>0$, continuity implies that
$ \qdown_{V}(f(\tau-\epsilon)) > \qdown_S(f(\tau-\epsilon)) $.
But, by definition of $\tau$, we have $f(\tau-\epsilon) \in \cS$.
This contradicts \Lemma{breve-prop}.
\end{proof}

Consider rounding the coordinate $z_i$.
Let $\vec{z'} \rightarrow \vec{z}$.
If $\frac{\partial \qdown_{V} }{ \partial z_i }(\vec{p}(\vec{z})) \geq 0$
then set $z'_i \leftarrow 1$, otherwise set $z'_i \leftarrow 0$.
Then, by multilinearity, for all $\vec{y}$ on the line segment between $\vec{z}$ and $\vec{z'}$
we have $\qdown_{V}(\vec{y}) \geq \qdown_{V}(\vec{z}) > 0$, and hence $\vec{y} \in \cR$.
\Lemma{continuous} now implies that $\vec{p}(\vec{z'}) \in \cS$.

\subsection{Impossibility of rounding within $\cL$}
\label{app:Lrounding}

Consider the scenario $\Omega = \set{0,1}^{15}$ with $\vec{z} = (1/2,\ldots,1/2)$
so that $\mu_{\vec{z}}$ is the uniform distribution.
Define the events
\begin{align*}
\cE_1 &~=~ \setst{ \omega }{ \smallsum{i=1}{6} \: \omega_i \in \set{0,2,6} } \\
\cE_2 &~=~ \setst{ \omega }{ \omega_6=\omega_7=\omega_8=0 } \\
\cE_3 &~=~ \setst{ \omega }{ \omega_8=\omega_9=\omega_{10}=1 } \\
\cE_4 &~=~ \setst{ \omega }{ \smallsum{i=10}{15} \: \omega_i \in \set{0,2,6} }.
\end{align*}
The dependency graph $G$ for these events is the path graph $1$-$2$-$3$-$4$.

Recalling that $p_i = \mu_{\vec{z}}(\cE_i)$, we have
$$\vec{p} ~=~ \Big(\frac{17}{64},~ \frac{1}{8},~ \frac{1}{8},~ \frac{17}{64} \Big).$$
One may verify that $\vec{p} \in \cL$ using the vector
$$\vec{x} ~=~ \Big(\frac{4}{10},~ \frac{3}{10},~ \frac{3}{10},~ \frac{4}{10} \Big).$$

Suppose we round $z_8$ to $0$. (The other case is symmetric.)
The probability vector is now
$$\vec{p}' ~=~ \Big(\frac{17}{64},~ \frac{1}{4},~ 0,~ \frac{17}{64} \Big).$$
We claim that $\vec{p'} \not\in \cL$.
To see this, note that $p_3 = 0$, so we may delete vertex $3$, after which the dependency graph has
the single edge $1$-$2$.
For a graph consisting of a single edge, the region $\cL$ can be shown to be precisely
$ \setst{ (p_1,p_2) }{ \sqrt{p_1} + \sqrt{p_2} \leq 1 } $.
Since this condition is violated for $\vec{p}'$, we have $\vec{p'} \not\in \cL$.

Thus $z_8$ cannot be rounded either to $0$ or to $1$ while preserving membership in $\cL$.


\section{Hardness of evaluation and deciding membership}
\label{app:hardness}

In this section we complement our positive results with some negative
ones. \ifbool{twocol}{%
  The first such hardness results show that exact evaluation of $\qdown_V$
  and exact decision of membership in Shearer region are both
  \SharpPHard: due to lack of space these results are consigned to the
  full version and we focus here on results about approximation.
  In \Section{approxEvalMemb} we show that approximation within
  exponentially small error is computationally hard, and then deduce
  that algorithms for approximating $\qdown_V$ must have runtime that
  depends on the slack. Finally, we can also show a positive result
  for a relaxation of the membership problem --  that one can
  efficiently decide membership in the region for the original LLL,
  which is a strict subset of Shearer's region -- but due to lack of
  space, this result is also not included here and may be found in the full version of the paper.
 }{ %
  First, in \Section{exactEvalMemb} we show that exactly evaluating
  $\qdown_V$ and deciding membership is \SharpPHard.  Then,
  in \Section{approxEvalMemb} we show similar results with
  exponentially small error, and show that algorithms for
  approximating $\qdown_V$ must have runtime that depends on the
  slack.
Finally, in \Section{LLLregion} we show a positive result: that one can
efficiently decide membership in the region for the original LLL,
which is a strict subset of Shearer's region.
} %

\ifbool{twocol}{}{
\subsection{Exact evaluation and membership}
\SectionName{exactEvalMemb}

Our starting point is the following known hardness result.

\begin{theorem}[\cite{DagumLuby92}]
\label{thm:matching-hardness}
For a given 3-regular bipartite graph, it is \SharpPHard to compute the number of perfect matchings.
\end{theorem}

\vspace{6pt}

From here, we obtain by a standard reduction the hardness of computing the
alternating-sign independence polynomial.
In the following, to emphasize the graph under consideration,
we deviate from our previous notation slightly by letting
$\breve{q}_G(p) = \sum_{I \in \Ind(V)} (-1)^{|I|} p^I$, where $G=(V,E)$.

\begin{theorem}
\label{thm:eval-hardness}
For a 4-regular graph $G=(V,E)$ and $|V| < k < |V|^2$ given on the input, it is \SharpPHard to compute $\breve{q}_G(1/k,\ldots,1/k)$.
\end{theorem}

\vspace{6pt}

We note that for a $4$-regular graph, the Shearer region is known to contain at least the line
segment between $(0,\ldots,0)$ and $(\frac{1}{4e}, \ldots, \frac{1}{4e})$. So we claim that it is
\SharpPHard to evaluate Shearer's polynomial even on points that are inside the Shearer region with a large slack.

\begin{proof}
Let $H$ be a given 3-regular bipartite graph on $n+n$ vertices. We define $G = (V,E)$ to be the line graph of $H$, which is 4-regular. We have $|V| = 3n$, the number of edges of $H$. Independent sets in $G$ correspond to matchings in $H$. We have
\ifbool{twocol}{
  \begin{align*}
    \breve{q}_G(1/k,\ldots,1/k)
    &= \sum_{I \in \Ind(G)} \left(
    -\frac{1}{k} \right)^{|I|}\\
    &= \sum_{\text{matching }M \subset H}
    \left( -\frac{1}{k} \right)^{|M|}\\
    &= \frac{1}{k^n}
    \sum_{\text{matching }M \subset H} (-1)^{|M|} k^{n-|M|}.
  \end{align*}
}{
  $$ \breve{q}_G(1/k,\ldots,1/k) = \sum_{I \in \Ind(G)} \left( -\frac{1}{k} \right)^{|I|}
= \sum_{\text{matching }M \subset H} \left( -\frac{1}{k} \right)^{|M|}
= \frac{1}{k^n} \sum_{\text{matching }M \subset H} (-1)^{|M|}
k^{n-|M|}.$$
} Let us denote $b_k = \sum_{\text{matching }M \subset H} (-1)^{|M|} k^{n-|M|}$, which is an integer.
Perfect matchings in $H$ have cardinality $n$.  Therefore, each non-perfect matching contributes a multiple of $k$ here, only perfect matchings contribute $1$ (with a sign depending on the parity on $n$; assume wlog that $n$ is even). Hence we have $b_k = \# \text{ perfect matchings} \pmod k$.
If we could compute $b_k$, say for any $|V| < k < |V|^2$, then we could recover the number of perfect matchings in $H$ by the Chinese remainder theorem with $|V|$ choices of prime numbers $k$, $|V| < k < |V|^2$ (which exist for large enough $|V|$ by the prime number theorem), since the number of perfect matchings is upper-bounded by $n! < |V|^{|V|}$.
This proves that computing $\breve{q}_G(1/k,\ldots,1/k) = b_k / k^n$ for $|V| < k < |V|^2$ is
\SharpPHard.
\end{proof}

Next, we show that for an unrestricted point $p$, it is \SharpPHard even to compute the sign of $\breve{q}_G(p)$.

\begin{theorem}
\label{thm:sign-hardness}
For a graph $G = (V,E)$ and rational $p \in [0,1]^V$ given on the input, it is \SharpPHard to decide whether $\breve{q}_G(p) > 0$.
\end{theorem}

\begin{proof}
Let $G$ be a given graph as in the proof of Theorem~\ref{thm:eval-hardness}, $|V| = 3n$ and $|V| < k < |V|^2$. Define $G'$ to be a graph obtained from $G$ by adding a new vertex $z$ and adding edges between $z$ and all the vertices of $G$. We have
$$\breve{q}_{G'}(1/k,\ldots,1/k,p_z) = \breve{q}_G(1/k,\ldots,1/k) - p_z $$
because the only independent set in $G'$ containing $z$ is $\{z\}$. We can also assume that $\breve{q}_G(1/k,\ldots,1/k) = b_k / k^n$ for some integer $b_k$, as in the proof of Theorem~\ref{thm:eval-hardness}. The possible range for $b_k$ is $[-(8k)^n, (8k)^n]$, since in the proof of Theorem~\ref{thm:eval-hardness}, $|b_k| \leq k^n |\Ind(G)| \leq k^n 2^{3n}$.

Suppose that we can decide whether $\breve{q}_{G'}(1/k,\ldots,1/k,p_z) > 0$ for a given $p_z = b' /
k^n$. That is, we can decide whether $\breve{q}_G(1/k,\ldots,1/k) > p_z$. Then we can compute
$\breve{q}_G(1/k,\ldots,1/k) = b_k / k^n$ by a binary search on $p_z = b' / k^n$. Since we have
$2(8k)^n$ possible values for $b'$, the binary search takes $1 + n \log_2 (8k) \leq 1 + n \log_2
(72n^2)$ steps. Therefore we could compute the value of $\breve{q}_G(1/k,\ldots,1/k)$, which is
\SharpPHard.
\end{proof}

The same argument also gives the following.

\begin{theorem}
\label{thm:membership-hardness}
For a graph $G = (V,E)$ and rational $p \in [0,1]^V$ given on the input, it is \SharpPHard to decide whether $p$ is in the Shearer region.
\end{theorem}

\begin{proof}
Let $\cS_G$ denote the Shearer region for a graph $G$. Let $G' = G+z$ be a graph as in the proof of
Theorem~\ref{thm:sign-hardness}, with all edges between $z$ and $G$. For a given $p_z \in [0,1]$ and
$|V| < k < |V|^2$, we would like to decide whether $\breve{q}_{G'}(1/k,\ldots,1/k,p_z) =
\breve{q}_G(1/k,\ldots,1/k) - p_z > 0$. As above, we know that $(1/k,\ldots,1/k) \in \cS_G$, which means that $0 < \breve{q}_G(1/k,\ldots,1/k) < 1$. Therefore, as $p_z$ varies from $0$ to $1$, $\breve{q}_{G'}(1/k,\ldots,1/k,p_z) = \breve{q}_G(1/k,\ldots,1/k) - p_z$ decreases from a positive value to a negative one.

We use the following characterization: $p \in \cS_{G'}$ if and only if there is a continuous
path from the origin to $p$ such that $\breve{q}_{G'}(x) > 0$ for each point $x$ on the path
\cite[Theorem 2.10]{ScottSokal}. Here, we know that $(1/k,\ldots,1/k,0) \in \cS_{G'}$ and
$\breve{q}_{G'}(1/k,\ldots,1/k,p_z)$ is decreasing in $p_z$; therefore checking whether
$(1/k,\ldots,1/k,p_z) \in \cS_{G'}$ is equivalent to checking whether
$\breve{q}_{G'}(1/k,\ldots,1/k,p_z) > 0$, which is \SharpPHard.
\end{proof}
}

\subsection{Approximate evaluation and membership}
\SectionName{approxEvalMemb}

\ifbool{twocol}{}{Perhaps a more interesting question is how accurately
  we can evaluate $\breve{q}_G(p)$ or check membership in the Shearer
  region, when errors are allowed. }As our main positive result shows,
$\breve{q}_G(p)$ for $p$ well inside the Shearer region (with constant
slack) can indeed be evaluated {\em approximately}, within
polynomially small error. \ifbool{twocol}{The proofs of some of these
  results are not included here due to lack of space and may be found
  in the full version. }{ Our hardness reductions here show that certain exponentially small errors are not achievable. We obtain the following results automatically, from the fact that the possible values of $\breve{q}_G(p)$ in our reduction are integer multiples of $1/k^n$, where $|V| = 3n$ and $k < |V|^2$, so $1/k^n > 1/|V|^{2|V|/3}$.}

\begin{theorem}
\label{thm:app-eval-inapprox}
For a 4-regular graph $G=(V,E)$ and $|V| < k < |V|^2$ given on the input, it is \SharpPHard to
compute $\breve{q}_G(1/k,\ldots,1/k)$ within an additive error of $1/ (2k^{|V|/3})$.
\end{theorem}

\begin{theorem}
\label{thm:sign-inapprox}
For a graph $G = (V,E)$ and rational $p \in [0,1]^V$ given on the input, it is \SharpPHard to distinguish whether $\breve{q}_G(p) \geq 1/|V|^{|V|}$ or $\breve{q}_G(p) \leq 0$.
\end{theorem}
\ifbool{twocol}{}{
\vspace{6pt}

With a slight extension of the above proof for membership hardness, we get the following.
}
\begin{theorem}
\label{thm:app-membership-inapprox}
For a graph $G = (V,E)$ and rational $(p_1,\ldots,p_n) \in [0,1]^V$ given on the input, it is
\SharpPHard to distinguish between $(p_1+\epsilon,\ldots,p_n + \epsilon) \in \cS_G$ and
$(p_1,\ldots,p_n) \notin \cS_G$, for $\epsilon = 1/|V|^{|V|}$.
\end{theorem}
\ifbool{twocol}{}{
\begin{proof}
Consider the reduction we used in the proof of Theorem~\ref{thm:sign-hardness}.
It shows that it is \SharpPHard to distinguish whether $\breve{q}_G(1/k,\ldots,1/k) \geq b' / k^n$ or $\breve{q}_G(1/k,\ldots,1/k) \leq (b'-1) / k^n$, for given $b' > 0$; here, $|V| = 3n$. We let $p_z = (b'-1) / k^n$ and consider the graph $G' = G+z$. In the first case, $(1/k,\ldots,1/k,p_z)$ is in the Shearer region of $G'$ while in the second case it is not.

In the first case, when $\breve{q}_G(1/k,\ldots,1/k) \geq b' / k^n$, we consider a modified point
$(1/k + \epsilon, \ldots, 1/k + \epsilon, p_z + \epsilon)$ where $\epsilon = 1/|V|^{|V|}$. By the
convexity of $\breve{q}_G(\lambda p)$ in $\lambda$ (see \cite{HV15}), we have
\ifbool{twocol}{
  \begin{multline*}
    \breve{q}_G(1/k+\epsilon, \ldots, 1/k+\epsilon)\\
    \geq 1 - \frac{1/k+\epsilon}{1/k} (1-\breve{q}_G(1/k,\ldots,1/k)) \geq \frac{b'}{k^n} - k \epsilon.
  \end{multline*}
}{
  $$\breve{q}_G(1/k+\epsilon, \ldots, 1/k+\epsilon) \geq 1 -
  \frac{1/k+\epsilon}{1/k} (1-\breve{q}_G(1/k,\ldots,1/k)) \geq
  \frac{b'}{k^n} - k \epsilon.$$
  }
Furthermore,
\ifbool{twocol}{
  \begin{multline*}
    \breve{q}_{G'}(1/k+\epsilon, \ldots, 1/k+\epsilon, p_z+\epsilon)
    \\=
    \breve{q}_G(1/k+\epsilon,\ldots,1/k+\epsilon) - p_z - \epsilon
    \\\geq \frac{b'}{k^n} - k \epsilon - \frac{b'-1}{k^n} - \epsilon > 0
  \end{multline*}

}{$$\breve{q}_{G'}(1/k+\epsilon, \ldots, 1/k+\epsilon, p_z+\epsilon) = \breve{q}_G(1/k+\epsilon,\ldots,1/k+\epsilon) - p_z - \epsilon \geq \frac{b'}{k^n} - k \epsilon - \frac{b'-1}{k^n} - \epsilon > 0$$}
since $\epsilon = {1}/{|V|^{|V|}} < {1}/{k^{3n/2}}$.
Moreover, there is a path from the origin to $(1/k+\epsilon,\ldots,1/k+\epsilon,p_z+\epsilon)$ where
$\breve{q}_{G'}$ is positive, which means that $(1/k+\epsilon,\ldots,1/k+\epsilon,p_z+\epsilon) \in
\cS_{G'}$.

In the second case, when $\breve{q}_G(1/k,\ldots,1/k) \leq (b'-1) / k^n$, we have
$$ \breve{q}_{G'}(1/k,\ldots,1/k, p_z) = \breve{q}_G(1/k,\ldots,1/k) - p_z \leq 0.$$
Here, $(1/k,\ldots,1/k,p_z) \notin \cS_{G'}$.
Therefore, distinguishing between these two cases would allow us to solve a \SharpPHard problem.
\end{proof}
}
It remains open whether membership in the Shearer region is polynomially checkable within {\em polynomially small error}.

\medskip

Next, we use Theorem~\ref{thm:app-membership-inapprox} to prove that it is in fact \SharpPHard to approximate the independence polynomial even within {\em polynomially large factors}, when $\vec{p}$ is inside but close to the boundary of the Shearer region.

\begin{theorem}
\label{thm:app-poly-factor-hardness}
For a graph $G = (V,E)$, $|V|=n$, and rational $\vec{z} \in [0,1]^V$ given on the input, such that
$(1 + \frac{1}{n^{2n}}) \vec{z} \in \mathcal{S}_G$, it is \SharpPHard to approximate
$\qdown_V(\vec{z})$ with any $\poly{n}$ factor.
\end{theorem}

\begin{proof}
Suppose that given $G$, $\vec{z}$ as above, we can compute a number $\breve{Q}_V$ such that
$\qdown_V(\vec{z}) \leq \breve{Q}_V \leq n^c \qdown_V(\vec{z})$, for some absolute constant $c>0$.
Then clearly we can also do this for $\qdown_S(\vec{z})$, $S \subseteq V$, by considering the
subgraph induced by $S$. Suppose also that $n$ is sufficiently large, say $n \geq 2c+2$.
We claim that then by a polynomial number of calls to such an algorithm, we
can distinguish for a given point $\vec{p} \in [0,1]^V$ whether $\vec{p} + \frac{1}{n^n} \vec{1} \in
\mathcal{S}_G$ or $\vec{p} \notin \mathcal{S}_G$, which is a \SharpPHard problem by
Theorem~\ref{thm:app-membership-inapprox}.

Let $\phi(t) = \qdown_V(t \vec{p})$. Clearly $\phi(0) = 1$, and it was shown in \cite{HV15} that
$\phi$ is convex and decreasing. We aim to find the minimum $t>0$ such that $\phi(t) = 0$, which defines the nearest point on the boundary of $\mathcal{S}_G$ in the direction of $\vec{p}$. We can assume that $\sum_{i=1}^{n} p_i \geq 1$, otherwise $\vec{p} \in \mathcal{S}_G$ trivially.
We use the following algorithm: We start with $t = 0$. Given $t$, we estimate $\phi(t)$ and $\phi'(t)$ (within polynomial factors) using the assumed algorithm. This can be done, since $\phi(t) = \qdown_V(t \vec{p})$, and
$$ \phi'(t) = \frac{d}{dt} \qdown_V(t \vec{p})
 = \sum_{i=1}^{n} p_i \frac{\partial}{\partial{z_i}} \qdown_V(\vec{z}) \big|_{t \vec{p}}
 = -\sum_{i=1}^{n} p_i \qdown_{V \setminus \Gamma^+(i)}(t \vec{p}).$$
We will show that we only apply this computation to points $t$ such that $(1 + \frac{1}{n^{2n}}) t \vec{p} \in \mathcal{S}$. For such points $\phi(t) > 0$, $\phi'(t) < 0$ and we can also estimate $\frac{\phi(t)}{|\phi'(t)|}$. Let $D(t)$ be our estimate, such that $n^{-2c} \frac{\phi(t)}{|\phi'(t)|} \leq D(t) \leq \frac{\phi(t)}{|\phi'(t)|}$.
Given this estimate, we replace $t$ by $t' = t + \frac12 D(t)$. We
repeat this process as long as $D(t) \geq 1/n^{n+1+2c}$ and $t < 1$.
If we reach $t \geq 1$, we answer YES; else if $D(t)$ drops below $1/n^{n+1+2c}$, we answer NO.

We note the following: Assuming that the minimum positive root of
$\phi$ is $\xi_0$ and $0 \leq t < \xi_0$, we have $t + D(t) \leq t +
\frac{\phi(t)}{|\phi'(t)|} \leq \xi_0$ by convexity of
$\phi$. Therefore, the additive slack at any point $t$ is at least
$D(t)$. Since we update the point to $t' = t + \frac12 D(t)$, we
always retain slack at least $\frac12 D(t)$, which is guaranteed to be
at least $\frac{1}{2n^{n+1+2c}} \geq \frac{1}{n^{2n}}$ (for $n \geq 2c+2$), otherwise we terminate.
This proves the above claim that we only evaluate at points $t$
such that $(1 + \frac{1}{n^{2n}}) t \vec{p} \in \mathcal{S}$.

On the other hand, if $\delta \defeq \xi_0 - t$, we have
$\qdown_V(t \vec{p} + \delta p_i \vec{e}_i) \geq 0$ for $1 \leq i \leq
n$ since $(t+\delta)p$ is at the boundary of $\mathcal{S}_G$. We then have
\ifbool{twocol}{
  \begin{align*}
    \phi(t + \delta) - \phi(t)
    &\leq \min_{1 \leq i \leq n} \qdown_V(t \vec{p} + \delta p_i \vec{e}_i) - \qdown_V(t \vec{p})\\
    &=  -\max_{1 \leq i \leq n} \delta p_i \frac{\partial \qdown_V}{\partial z_i} \big|_{\vec{z} = t \vec{p}}\\
    &\leq - \frac{\delta}{n} \sum_{i=1}^{n} p_i \frac{\partial \qdown_V}{\partial z_i} \big|_{\vec{z} = t \vec{p}}
  = \frac{\delta}{n} \phi'(t).
  \end{align*}

}{
  $$ \phi(t + \delta) - \phi(t) \leq \min_{1 \leq i \leq n} \qdown_V(t \vec{p} + \delta p_i \vec{e}_i) - \qdown_V(t \vec{p})
 =  -\max_{1 \leq i \leq n} \delta p_i \frac{\partial \qdown_V}{\partial z_i} \big|_{\vec{z} = t \vec{p}}
  \leq - \frac{\delta}{n} \sum_{i=1}^{n} p_i \frac{\partial \qdown_V}{\partial z_i} \big|_{\vec{z} = t \vec{p}}
  = \frac{\delta}{n} \phi'(t).
  $$
}
Therefore, since $\phi(t + \delta) = 0$, we get $\frac{\phi(t)}{|\phi'(t)|} \geq \frac{\delta}{n} = \frac{\xi_0 - t}{n}$. By our approximation guarantee, $D(t) \geq n^{-2c} \frac{\phi(t)}{|\phi'(t)|}
\geq n^{-2c-1} (\xi_0-t)$.  (Note that this also means that $(t + n^{1+2c} D(t))\vec{p} \not\in \mathcal{S}$.)
So when we replace $t$ by $t + \frac12 D(t)$, we decrease the distance
$\xi_0 - t$ to the nearest root by a factor of $1 - 1/(2n^{2c+1})$ in the worst case. After $2 n^{2c+1} (n+2+2c) \log n$ steps, the distance decreases by a factor of
\[
    \Big(1 - \frac{1}{2n^{2c+1}} \Big)^{2 n^{2c+1} (n+2+2c) \log n} ~<~ \frac{1}{n^{n+2+2c}}.
\]
Initially, we have $\xi_0 \leq n$ because $p_i \geq 1/n$ for some $i \in [n]$. Hence,
the quantity $\xi_0 - t$ as well as $D(t)$ must shrink below $1/n^{n+1+2c}$ in a polynomial number of steps.

If we terminate because $t \geq 1$, we have certified that $\vec{p} \in \mathcal{S}$ and we can answer YES.
If we terminate because $D(t) < 1/n^{n+1+2c}$ then it is the case that $t < 1$, and we know that
$(t + n^{1+2c} D(t)) \vec{p} \notin \mathcal{S}$.
Hence $(1 + 1/n^{n}) \vec{p} \notin \mathcal{S}$ and we can answer NO.
\end{proof}

\begin{corollary}[restatement of \Theorem{ShearerHardnessIntro}]
\label{cor:estimateQ}
If there is an algorithm to estimate $\qdown_V(\vec{p})$ to within a $\poly{n}$ multiplicative
factor, assuming that $(1+\alpha) \vec{p} \in \mathcal{S}$, and running in time $(n \log
\frac{1}{\alpha})^{O(\log n)}$, then $\text{\#P} \subseteq \DTIME(n^{O(\log n)})$.
\end{corollary}

\begin{proof}
Suppose we have such an algorithm. Then we can run it for $\alpha = \frac{1}{n^{2n}}$, and solve a
\SharpPHard problem (from Theorem~\ref{thm:app-poly-factor-hardness}) in running time $(n \log \frac{1}{\alpha})^{O(\log n)} = n^{O(\log n)}$.
\end{proof}

In contrast, the running time of our algorithm (for a constant-factor
approximation) is $n^{O(\frac{1}{\alpha} \log d)}$. Again, there is an
open question here, whether there is an approximation algorithm
(possibly even an FPTAS) under these conditions with running time at most $\text{quasi-poly}(n,1/\alpha)$.

\ifbool{twocol}{}{
\subsection{Membership in the original LLL region}
\SectionName{LLLregion}

The original statement of the Lov\'asz Local Lemma \cite{ErdosLovasz,Spencer77}
had a stronger hypothesis than Shearer's formulation.
It stated that $\mu(\land_{i=1}^n\lnot\cE_i) > 0$ if $G$ is a dependency graph
for events $\cE_1,\ldots,\cE_n$ and if $\vec{p}$ lies in the set
\ifbool{twocol}{
  \begin{multline*}
    \cL_G \defeq \Big\{~ \vec{p} \in [0,1]^V ~:~ \exists \vec{x} \in (0,1)^V ~\text{s.t.}~\\
    p_i \leq x_i \cdot \smallprod{(i,j) \in E}{} \, (1-x_j) ~\:\forall i \in V ~\Big\}.
  \end{multline*}
}{
  $$
\cL_G \defeq \Big\{~ \vec{p} \in [0,1]^V ~:~ \exists \vec{x} \in (0,1)^V ~\text{s.t.}~
    p_i \leq x_i \cdot \smallprod{(i,j) \in E}{} \, (1-x_j) ~\:\forall i \in V ~\Big\}.
    $$
  }
Shearer's results imply that $\cL_G \subseteq \cS_G$
(see also \cite{ScottSokal,HV15}).
Interestingly,
although deciding membership in $\cS_G$ is \SharpPHard,
membership in $\cL_G$ can be decided in polynomial time within exponentially small errors.

\begin{theorem}
For a given graph $G = (V,E)$, $|V|=n$, rational $(p_1,\ldots,p_n) \in [0,1]^V$ and $\epsilon>0$, we
can distinguish between $(p_1+\epsilon,\ldots,p_n+\epsilon) \in \cL_G$ and $(p_1,\ldots,p_n) \notin
\cL_G$ , in time $\text{poly}(n, \log \frac{1}{\epsilon})$.
\end{theorem}

\begin{proof}
By taking logs, we can write equivalently
\ifbool{twocol}{
  \begin{multline*}
    \cL_G ~=~ \Big\{~ \vec{p} \in [0,1]^V ~:~ \exists \vec{x} \in (0,1)^V ~\text{s.t.}~\\
    \log p_i \leq \log x_i + \smallsum{(i,j) \in E}{} \log (1-x_j) ~\:\forall i \in V ~\Big\}.
  \end{multline*}
}{
  $$
\cL_G ~=~ \Big\{~ \vec{p} \in [0,1]^V ~:~ \exists \vec{x} \in (0,1)^V ~\text{s.t.}~
    \log p_i \leq \log x_i + \smallsum{(i,j) \in E}{} \log (1-x_j) ~\:\forall i \in V ~\Big\}.
 $$} 
Thus $p \in \cL_G$ is equivalent to the following set being nonempty:
\ifbool{twocol}{
  \begin{multline*}
    \cX_{G,p} ~=~ \Big\{~ \vec{x} \in (0,1)^V ~:~\\
    \log p_i \leq \log x_i + \smallsum{(i,j) \in E}{} \log (1-x_j) ~\:\forall i \in V ~\Big\}.
  \end{multline*}
}{$$
\cX_{G,p} ~=~ \Big\{~ \vec{x} \in (0,1)^V ~:~
    \log p_i \leq \log x_i + \smallsum{(i,j) \in E}{} \log (1-x_j) ~\:\forall i \in V ~\Big\}.
$$}
Note that this is a convex set: $p_i$ is fixed here, and $\phi_i(x) = \log x_i + \sum_{(i,j) \in E}
\log (1-x_j)$ is a concave function of $x \in (0,1)^V$. Also, it is easy to implement a separation
oracle for $\cX_{G,p}$:
Given a point $x$, we can check directly if all the constraints are satisfied, and if not we can compute a separating hyperplane whose normal vector is the gradient of $\phi_i(x)$.

Suppose now that $p + \epsilon = (p_1+\epsilon, \ldots, p_n + \epsilon) \in \cL_G$. Let $x \in (0,1)^V$ be such that $p_i + \epsilon \leq  x_i \prod_{(i,j) \in E} (1-x_j)$. Clearly, $x_i \geq \epsilon$. Also, for any $\xi_i \in [0,\epsilon]$,
\ifbool{twocol}{
  \begin{align*}
    (x_i - \xi_i) \prod_{(i,j) \in E} (1-(x_j - \xi_j))
    &\geq (x_i - \xi_i) \prod_{(i,j) \in E} (1-x_j)\\
    &\geq x_i \prod_{(i,j) \in E} (1-x_j) - \xi_i \geq p_i.
  \end{align*}
}{$$ (x_i - \xi_i) \prod_{(i,j) \in E} (1-(x_j - \xi_j)) \geq (x_i - \xi_i) \prod_{(i,j) \in E} (1-x_j)
 \geq x_i \prod_{(i,j) \in E} (1-x_j) - \xi_i \geq p_i.$$}
This means that the box $[x - \epsilon, x]$ is contained in $\cX_{G,p}$.
The volume of this box is $\epsilon^n$, while $\cX_{G,p}$ is contained in the box $[0,1]^V$,
of volume $1$. Therefore, by the ellipsoid method, we can find a point in $\cX_{G,p}$ in
$\text{poly}(n,\log \frac{1}{\epsilon})$ iterations, which certifies that $p \in \cL_G$ and we can
answer YES. If the ellipsoid method fails to find such a point, it must be the case that $p+\epsilon
\notin \cL_G$, in which case we can answer NO.
\end{proof}

In particular, in $\text{poly}(n)$ time we can decide about membership in the LLL region within a
$1/n^n$ additive error, which is \SharpPHard for the Shearer region.
}
%

%


\section{Extension to graphs of bounded connective constant}
\label{sec:extens-graphs-bound}

The connective constant, first studied by
Hammersley~\cite{hammersley_percolation_1957}, is a natural notion of
the average degree of a graph.  The definition is best motivated in
the setting of infinite regular lattices (e.g., $\Z^2$), though it
extends easily to general graph families.  Note that the maximum and
average degrees of $\Z^2$ are both $4$, and in this respect it is not
distinguishable from the infinite $4$-regular tree.  However, it is
clear that $\Z^2$ is very different from the regular tree (in
particular due to its small girth), and the connective constant may be
seen as a notion of average degree that tries to capture this
difference.

For a fixed vertex $v$ in $\Z^2$, consider $N(v, \ell)$, the number of
\emph{self-avoiding walks} in the lattice starting at $v$. (In the
special case of the lattice, this number depends only upon $\ell$ and
not $v$).  We then have
\begin{displaymath}
  2^\ell < N(v, \ell) < 3^\ell.
\end{displaymath}
The connective constant measures the rate of growth of $N(v, \ell)$ as
a function of $\ell$.  Formally, the connective constant
$\Delta(\Z^2)$ of $\Z^2$ is given by
\begin{displaymath}
  \Delta(\Z^2) = \lim_{\ell \rightarrow \infty} N(v, \ell)^{1/\ell}.
\end{displaymath}
(The limit on the right hand side above can be shown to exist in the
case of $\Z^2$ and other regular lattices; see,
e.g., \cite{madras96:_self_avoid_walk}.  However, computing the exact
value of the connective constant is an open problem for most regular
lattices, with one celebrated
exception~\cite{duminil-copin_connective_2012}.)  For an infinite
family of finite graphs, we may similarly define the connective
constant as follows:

\begin{definition}[\textbf{Connective constant: finite
    graphs~\cite{SSSY15}}]
  Let $\mathcal{G}$ be an infinite family of finite graphs.  We say
  that the connective constant of graphs in $\mathcal{G}$ is at most
  $\Delta$ if there exist positive constants $a$ and $c$ such that for
  any $G \in \mathcal{G}$ with at least $n$ vertices, any
  $\ell \geq a \log n$ and any vertex $v$ in $G$, the number
  $N(v, \ell)$ of self-avoiding walks in $G$ of length $\ell$ starting
  at $v$ is at most $c\Delta^\ell$.
\end{definition}

\vspace{6pt}

Note that the connective constant of graphs of maximum degree $d$ is
at most $d-1$.  However, as in the case of lattices, it can be much
smaller than this crude bound; in particular, it can be bounded
even when the maximum degree is unbounded.  An important example is
that of graphs sampled from the sparse Erd\H{o}s-R\'enyi random graph
model $\mathcal{G}(n, d/n)$.  When $d$ is a constant, the
connective constant of such graphs is at most $d$ w.h.p.  On the other
hand, the maximum degree of such a graph on $n$ vertices is
$\Omega\inp{\frac{\log n}{\log \log n}}$ w.h.p.

The connective constant turns out to have an important connection with
correlation decay based algorithms for the independence polynomial.
Recall that Weitz showed that when $0 \leq \lambda < \lambda_c(d)$,
there is an FPTAS for $Z_G(\lambda)$ on graphs of degree at most $d$.
In \cite{SSSY15}, this was extended to all graphs of connective
constant $d-1$, without any bound on the maximum degree. (Note that
graphs of maximum degree $d$ have connective constant at most $d-1$,
so this is a strict generalization of Weitz's result even in the
bounded degree setting.)

In the setting of complex activities, our main theorem,
\Cref{thm:mainShearerIntro}, also generalizes to graphs of bounded
connective constant.  The proof presented in
\Cref{sec:corr-decay-with-2} is already sufficient to establish this
extension with a few small modifications, which we now proceed to
describe.  In particular, we prove the following modification of
\Cref{thm:mainShearer}.

\begin{theorem}[\textbf{FPTAS for graphs of bounded connective constant}]
  \label{thm:mainShearer-connective}
  Let $\mathcal{G}$ be an infinite family of finite graphs with
  connective constant at most $\Delta$, and let the constant $a$ be as
  in the definition of the connective constant.  Given a graph
  $G = (V, E) \in \mathcal{G}$ on $n$ vertices, a parameter vector
  $\vec{p}$ such that $(1+\alpha)^2 \vec{p} \in \cS$, and a positive
  $\epsilon \leq 1/n$, define
  $\ell = \max\inb{a \log n, \frac{4}{\sqrt{\alpha}}
    \log\inp{\frac{n}{\epsilon\alpha}}}$.  Then a
  $\big( 1+O(\epsilon n) \big)$-approximation to $\qdown_V(\vec{p})$
  can be computed in time $O( n \Delta^\ell)$.
\end{theorem}

\begin{proof}[Proof (sketch)]
  The proof is very similar to that of \Cref{thm:mainShearer} in
  \Cref{sec:corr-decay-with-2} so we only describe the steps that need
  to be modified.  The first observation, already alluded to in the
  remark following \Cref{lem:recurrence}, is that the the size of the
  computation tree for computing $r_{V, v}$ up to depth $\ell$ is at
  most $N(v, \ell)$, the number of self-avoiding walks of length
  $\ell$ starting at $v$.  Since $G$ belongs to a graph family of
  connective constant at most $\Delta$, we have
  $N(v, \ell) = O(\Delta^\ell)$ when $\ell \geq a\log n$, so that the
  cost of expanding the tree to depth $\ell$ is at most
  $O(\Delta^\ell)$.  Thus, the total cost of computing each of the
  $R_{S_i, v_i}$ as in the proof of \Cref{thm:mainShearer} is
  $O(n\Delta^\ell)$.

  It remains to show that this achieves an $1 \pm O(\epsilon n)$ factor
  approximation to $\qdown_V$.  The proof is again similar, except
  that we now apply \Cref{cor:decayRoot} using $n$ as the bound $d$
  on  the maximum degree, and using the new definition of $\ell$ in the statement of
  above theorem (which also replaces the $d$ used in the definition of
  $\ell$ in \Cref{thm:mainShearer} by $n$).  With these two
  modification, we again obtain
  \begin{displaymath}
    \abs{\frac{\Xi_i}{\xi_i} - 1} \leq \epsilon \quad\text{for each $i$},
  \end{displaymath}
  and can complete the proof exactly as before.
\end{proof}

%
%
%
%


\section{Optimality of the decay rate}
\label{sec:optim-decay-rate}

We saw above that when the probability vector $\vec{p}$ input to our
approximation algorithm~(\Cref{alg:main}) for the independence
polynomial has slack $\alpha$ (i.e., when
$(1+\alpha)\vec{p} \in \mathcal{S}$), the running time of the
algorithm is exponential in $1/\sqrt{\alpha}$.  In the applications
outlined in \Cref{sec:ShearerVar,sec:membership}, \Cref{alg:main} is
invoked with $\vec{p}$ that have slack $\Theta(1/n)$ (where $n$ is the
number of vertices in the graph), and the above exponential dependence
on $1/\sqrt{\alpha}$ leads to a sub-exponential dependence on $n$.
While this is qualitatively superior to exponential time algorithms
that would result from naive brute force methods, or even methods
based on approximate counting algorithms whose running times are
exponential in $1/\alpha$ (e.g., those in \cite{PatelRegts}), one might
still ask if it is possible to get a better dependence on $\alpha$ so
as to improve the running times obtained in our applications.

In this section, we present evidence that ``correlation-decay'' based
methods based on Weitz's framework cannot in fact break the
sub-exponential barrier.  To do this, we revisit
\Cref{thm:decayShearer}, which is the main ingredient in the
complexity analysis of \Cref{alg:main}.  Recall that
\Cref{thm:decayShearer} considers the effect of truncating the
exponentially large computation tree (generated by
\Cref{lem:recurrence}) at a finite depth $\ell$ by showing that the
tree-recurrence of \cref{eq:tree-recursion} causes an amortized
version of the error to decay by a factor $(1-\Omega(\sqrt{\alpha}))$
at each level of the computation tree.  The factor of $\sqrt{\alpha}$
in the running time of \Cref{alg:main} came precisely from this form
of the decay factor.  In particular, to show that one cannot do better
than a sub-exponential dependence on $1/\alpha$ in the running time in
this framework, we need to show that a decay rate better than
$(1-\Omega(\sqrt{\alpha}))$, where $c$ is an appropriate constant in
$(0, 1/2]$, cannot be obtained in the general case.

To establish this, we consider the case where the computation tree is
actually a $d$-ary tree.  In addition to being the simplest possible
example of a computation tree, the $d$-ary tree also arises as the
truncated computation tree when \Cref{alg:main} is applied to a vertex
in a locally tree like graph (e.g., a random vertex in a random
regular graph).

For a given activity $\lambda$, the tree recurrence on a $d$-ary tree
reduces to the following univariate recurrence:
\begin{displaymath}
  f(x) \defeq \frac{\lambda}{(1 - x)^d}.
\end{displaymath}
Since the $d$-ary tree has degree $d + 1$ at all vertices (except at
the root which has degree $d$), the radius of the univariate Shearer
region for it is $\lambda_c'(d+1) = \frac{d^d}{(d+1)^{d+1}}$.
Note that $f$ and all its derivatives are increasing functions of
$x \in (0, 1)$. It is also not hard to show that when
$0 < \lambda < \lambda_c'(d+1)$, this recurrence has two fixed
points $x^{\star} < x^{\dagger}$ in $(0, 1)$, such that
$\lim_{\ell \rightarrow \infty}f^\ell(x) = x^{\star}$.  Our goal now
is to show that when $\lambda$ has slack $\alpha$ (i.e.,
$\lambda = (1-\alpha)\lambda_c'(d+1)$), the rate of convergence
of this recurrence is no better than $\inp{1 - O\inp{\sqrt{\alpha}}}$.
In particular, we will show that
\begin{equation*}
  \abs{f^l(0) - x^{\star}} \geq \Omega\inp{(1 - O\inp{\sqrt{\alpha}}^l}.
\end{equation*}
More formally, we have the following theorem.
\begin{theorem}
  \label{thm:optimal-decay}
  Let $d$ be a large enough positive degree.  Let
  $\alpha \in (0, 1/8)$ be arbitrary, and set
  $\lambda = (1 - \alpha)\lambda_c'(d+1)$.  Then, there exist
  $c_o = c_0(\alpha, d), c_1 = c_1(\alpha, d)$ and
  $l_0 = l_0(\alpha, d)$ such that for all $l \geq l_0$,
  \begin{displaymath}
    \abs{f^l(0) - x^{\star}} \geq c_0 \inp{1 - c_1\sqrt{\alpha}}^{l - l_0}.
  \end{displaymath}
\end{theorem}

Before proving the theorem, we record some useful properties of the
recurrence $f$.  

\newcommand{\xst}[2][d]{\ensuremath{x^{\star}(#1, #2)}}
\begin{observation}
  \label{obv:f-incr}
  $f$ and all its derivatives are strictly increasing in $\lambda$ and
  $x$ when $x \in [0, 1)$.
\end{observation}
\begin{proof}
  This follows from the form of $f$.
\end{proof}

\begin{fact}
  \label{fct:prop-f}
  Given $\lambda > 0$ and $d \geq 3$, let $\xst{\lambda}$ be the
  smallest fixed point (if one exists) of $f$ in $(0, 1)$.  Then, for
  all $d \ge 3$, $\xst{\lambda_c'(d+1)} = \frac{1}{d + 1}$ and
  $f'(\xst{\lambda_c'(d+1)}) = 1$.  Further, $\xst{\lambda}$ and
  $f'(\xst{\lambda})$ are strictly increasing functions of $\lambda$
  for $\lambda < \lambda_c'(d+1)$.
\end{fact}
\begin{proof}
  Since $f$ is convex in $[0, 1)$ and satisfies both $f(0) > 0$ and
  $\lim_{t \uparrow 1} f(t) - t = \infty$, it follows that the
  equation $x = f(x)$ has at most two roots in $[0, 1)$.  Further,
  since $\lambda < \lambda_c'(d+1) \leq 1$, we also have that
  $f'(x) \leq 1$ at the smaller such root.  It then follows that if
  $f'(x) = 1$ at such a root $x$, then $x$ is the unique root of $f$
  in $[0, 1)$.  We then verify by direct calculation that
  $f_{d, \lambda_c'(d+1)}(1/(d+1)) = 1/(d+1)$ and
  $f'_{d, \lambda_c'(d+1)}(1/(d+1)) = 1$.

  Since $f$ is strictly increasing in $\lambda$, it then follows that
  $\xst{\lambda}$ is also a strictly increasing function of $\lambda$.
  Since $f'$ is strictly increasing in both $\lambda$ and $x$ the last
  claim also follows.
\end{proof}

We are now ready to prove \Cref{thm:optimal-decay}.
\begin{proof}[Proof of \Cref{thm:optimal-decay}]
  It follows from \Cref{fct:prop-f} that $f'(\xst{\lambda}) < 1$.  So
  let $\delta \in (0, 1)$ be such that
  $1 - \delta = f'(\xst{\lambda})$.  Since we also have
  $f(\xst{\lambda}) = \xst{\lambda}$, we can solve to get
  $\xst{\lambda} = \frac{1 - \delta}{d + 1 - \delta}$, and hence
  $\lambda = (1 - \delta)\cdot\frac{d^d}{(d + 1 - \delta)^{d + 1}}$.
  Since $\lambda = (1 - \alpha)\lambda_c'(d+1)$, we have
  \begin{equation}
    1 - \alpha = (1 - \delta) \inp{1 - \frac{\delta}{d + 1}}^{- (d +
      1)}.\label{eq:104}
  \end{equation}
  We now use the following inequalities which are valid for all large
  enough positive $d$:
  \ifbool{twocol}{
  \begin{gather}
    \forall x \in [0, 1], (1 - x)\inp{1 - \frac{x}{d + 1}}^{-(d+1)}
    \leq 1 - x^3\label{eq:105}\\
    \label{eq:103}\exists \mu > 0 \text{ such that } \forall x \in [0, 1/2],\\ (1 -
    x)\inp{1 - \frac{x}{d + 1}}^{-(d+1)} \leq 1 - (x/\mu)^2.\nonumber
  \end{gather}
  }{
  \begin{gather}
    \forall x \in [0, 1], (1 - x)\inp{1 - \frac{x}{d + 1}}^{-(d+1)}
    \leq 1 - x^3\label{eq:105}\\
    \exists \mu > 0 \text{ such that } \forall x \in [0, 1/2], (1 -
    x)\inp{1 - \frac{x}{d + 1}}^{-(d+1)} \leq 1 - (x/\mu)^2.\label{eq:103}
  \end{gather}
  }
  Inequality \eqref{eq:105} applied to \cref{eq:104} implies that
  $\delta \leq \alpha^{1/3}$ for all $\alpha \in (0, 1)$.  Thus, for
  $\alpha \in (0, 1/8)$, $\delta \in (0, 1/2)$.  We can thus apply
  inequality \eqref{eq:103} to \cref{eq:104} for such $\alpha$ to get
  $\delta \leq \mu \sqrt{\alpha}$ when $\alpha \in (0, 1/8)$, where
  $\mu = \mu(d)$ is the constant appearing in \cref{eq:103}.  Thus,
  $1 - \mu \sqrt{\alpha} \leq f'(\xst{\lambda}) \leq 1$. Further, by
  the continuity and monotonicity of $f'$, there exists an
  $\epsilon > 0$ such that
  \begin{equation}
    x \in [\xst{\lambda} - \epsilon, \xst{\lambda}]
    \quad\implies\quad
    f'(x) \geq 1 -  2\mu\sqrt{\alpha}.\label{eq:14}
  \end{equation}
  Now, $\inp{f^l(0)}_{l = 0}^\infty$ is a strictly increasing sequence
  lying strictly below the fixed point $\xst{\lambda}$, and further satisfies
  \ifbool{twocol}{
  \begin{align*}
    \xst{\lambda} - f^{l+1}(0)
    &= f(\xst{\lambda}) - f(f^l(0))\\
    &\leq (1-\delta)(\xst{\lambda} - f^{l}(0)).
  \end{align*}
  }{
  \begin{displaymath}
    \xst{\lambda} - f^{l+1}(0) = f(\xst{\lambda}) - f(f^l(0)) \leq
    (1-\delta)(\xst{\lambda} - f^{l}(0)).
  \end{displaymath}
  }
  It therefore follows that there is an $l_0$ such that for $l \geq
  l_0$,  $f^{l}(0) \in [\xst{\lambda} - \epsilon, \xst{\lambda}]$.
  For $l \leq l_0$ we therefore have
  \ifbool{twocol}{
  \begin{align*}
    \xst{\lambda} - f^{l+1}(0)
    &= f(\xst{\lambda}) - f(f^l(0))\\
    &\geq
    (1-2\mu\sqrt{\alpha})(\xst{\lambda} - f^{l}(0)),
  \end{align*}
  }{
  \begin{displaymath}
    \xst{\lambda} - f^{l+1}(0) = f(\xst{\lambda}) - f(f^l(0)) \geq
    (1-2\mu\sqrt{\alpha})(\xst{\lambda} - f^{l}(0)),
  \end{displaymath}
  }
  where the last inequality uses \cref{eq:14}.  The result now follows
  by an induction on $l$.
\end{proof}


%
\bibliographystyle{acm}
\end{document}